\documentclass[12pt, reqno]{amsart}

\usepackage[T1]{fontenc}
\usepackage[utf8]{inputenc}
\usepackage[english]{babel}

\usepackage{mathrsfs, csquotes, amsfonts,amssymb, amsmath,mathtools,amssymb, enumerate,enumitem, xcolor, dsfont,braket}

\theoremstyle{plain}
\newtheorem{theorem}{{Theorem}}[section]
\newtheorem*{theorem*}{{Theorem}}
\newtheorem{proposition}[theorem]{Proposition}

\newtheorem*{proposition*}{Proposition}
\newtheorem{corollary}[theorem]{Corollary}
\newtheorem*{corollary*}{Corollary}
\newtheorem{lemma}[theorem]{Lemma}

\newtheorem*{lemma*}{Lemma}
\theoremstyle{definition}

\newtheorem*{definition*}{Definition}
\theoremstyle{remark}
\newtheorem{remark}[theorem]{Remark}
\newtheorem{notation}[theorem]{Notation}

\makeatletter

\@addtoreset{equation}{section}
\makeatother

\usepackage{color}
\definecolor{bleu_sombre}{rgb}{0,0,0.6}
\definecolor{Bl}{rgb}{0,0,0.6}
\definecolor{rouge_sombre}{rgb}{0.8,0,0}
\definecolor{vert_sombre}{rgb}{0,0.6,0}
\usepackage[plainpages=false,colorlinks,linkcolor=bleu_sombre,
citecolor=rouge_sombre,urlcolor=vert_sombre,breaklinks]{hyperref}

\usepackage{pgf,tikz,pgfplots}
\definecolor{webblue}{rgb}{0.22,0.45,0.70}
\definecolor{webred}{rgb}{0.5, 0.09, 0.09}
\definecolor{zzttqq}{rgb}{0.6,0.2,0.}
\usetikzlibrary{arrows}

\renewcommand{\leq}{\leqslant}	
\renewcommand{\geq}{\geqslant}
\newcommand{\C}{\mathbb{C}}
\newcommand{\Z}{\mathbb{Z}}
\newcommand{\R}{\mathbb{R}}
\newcommand{\N}{\mathbb{N}}

\newcommand{\dd}{\mathrm{d}}

\makeatletter

\@addtoreset{equation}{section}
\makeatother

\numberwithin{equation}{section}

\title[]{Dirac bag model on thin domains:\\ a microlocal viewpoint}

\author[L. Le Treust]{Loïc Le Treust}
\email{loic.le-treust@univ-amu.fr}
\address[L. Le Treust]{Aix Marseille Univ, CNRS, I2M, Marseille, France}

\author[T. Ourmières-Bonafos]{Thomas Ourmières-Bonafos}
\email{thomas.ourmieres-bonafos@univ-amu.fr }
\address[T. Ourmières-Bonafos]{Aix Marseille Univ, CNRS, I2M, Marseille, France}

\author[N. Raymond]{Nicolas Raymond}
\email{nicolas.raymond@univ-angers.fr}
\address[N. Raymond]{Univ Angers, CNRS, LAREMA, Institut Universitaire de France,SFR MATHSTIC, F-49000 Angers, France}

\pgfplotsset{compat=1.16}

\begin{document}
	
\begin{abstract}
	The Dirac operator is considered on a bidimensional domain whose boundary carries the infinite mass boundary condition. The analysis is focused on the existence of discrete spectrum and on its asymptotic description in the thin width limit. We revisit a recent result by using an alternative view point inspired by microlocal analysis.
	
	\end{abstract}	
	\maketitle
	\tableofcontents
	
	\section{Motivation and main theorem}
	
	\subsection{Framework and motivation}
	We consider a smooth simple curve with an arc-length parametrization $ \gamma : s\in I \mapsto\gamma(s)\in\mathbb{R}^2$, where $I = \mathbb{R}$ or $I = \mathbb{R}\big/ \ell \mathbb{Z}$ for $\ell > 0$.  We consider the normal $\mathbf{n}=\gamma'^\perp$ such that $(\gamma',\mathbf{n})$ is a direct orthonormal basis and we consider, for all $(s,t)\in S_\epsilon:=I\times(-\epsilon,\epsilon)$,
	\begin{equation}\Gamma(s,t)=\gamma(s)+t\mathbf{n}(s)=x=(x_1,x_2)\,.
	\label{eqn:Gamma}
	\end{equation}
	We assume that, for $\epsilon$ small enough, $\Gamma$ is injective and we let $\Omega_\epsilon=\Gamma(S_\epsilon)$. When $I = \mathbb{R}$, $\Omega_\epsilon$ is a waveguide of width $2\epsilon$ and when $I  = \mathbb{R}\big/ \ell \mathbb{Z}$, it is the tubular neighborhood of the closed curve $\Gamma$ of width $2\epsilon$ (see Figure \ref{fig:waveguides}). The curvature of $\gamma$, represented by \( \kappa(s) \), is characterized by
\[
	\gamma^{\prime\prime}(s) = \kappa(s)\textbf{n}(s)
\]
for all \(s \in I\). Furthermore, when $\ell < +\infty$, we assume that the parametrization $\gamma$ is anti-clockwise {\it i.e.}, $\int_0^\ell \kappa(s) ds = 2\pi$.
	
	\begin{figure}[htbp]
    \centering
    \includegraphics[width=1\textwidth]{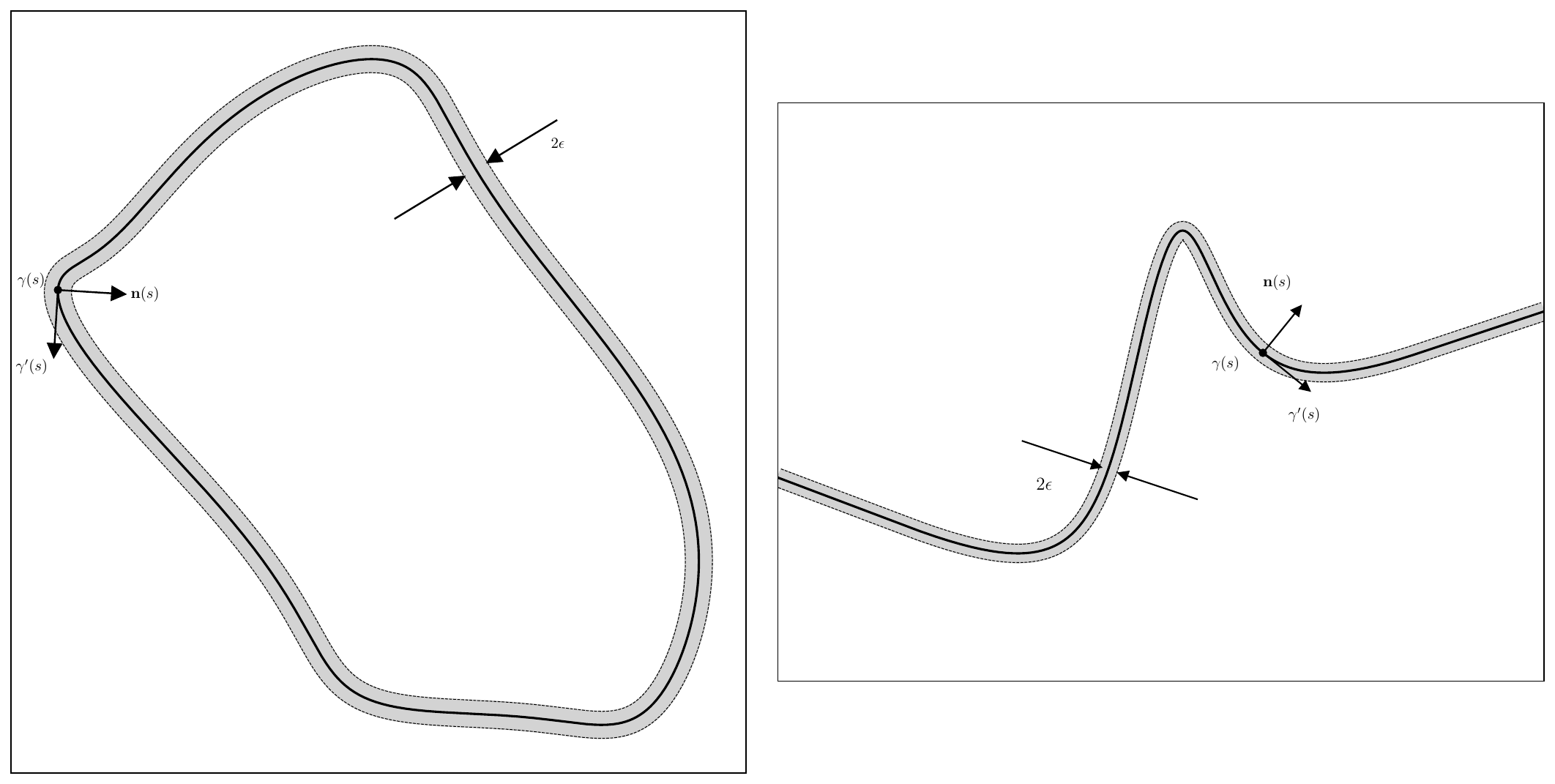}
    \caption{Closed and open waveguides}
    \label{fig:waveguides}
\end{figure}
We are interested in the Dirac operator
\begin{equation}\label{Def.Diracmain}
\mathscr{D}_\epsilon=\sigma_1D_{x_1}+\sigma_2D_{x_2}+m\sigma_3\,, \mbox{ with } D=-i\partial\,,\quad m\in \R\,,
\end{equation}	
equipped with the domain
\[\mathrm{Dom}(\mathscr{D}_\epsilon)=\{\psi\in H^1(\Omega_\epsilon,\mathbb{C}^2) : -i\sigma_3(\sigma\cdot \mathbf{N})\psi=\psi\,,\mbox{ on }\partial\Omega_\epsilon\}\,,\]	
	where $\mathbf{N}$ is the outward pointing normal to the boundary. The boundary conditions are the so-called "infinite mass boundary conditions". Note that the boundary $\partial\Omega_\epsilon$ has two connected components given by $\partial\Omega_\epsilon^\pm=\Gamma(I\times\{\pm\epsilon\})$. In the $+$ case, we have $\mathbf{N}(s)=\mathbf{n}(s)$ and in the $-$ case, we have $\mathbf{N}(s)=-\mathbf{n}(s)$. 
	
	The operator $(\mathscr{D}_\epsilon,\mathrm{Dom}(\mathscr{D}_\epsilon))$ is self-adjoint (see \cite[Ex. 4.20 \& Thm. 4.11]{BB16} or \cite{BFSVdB17,BBKOB22} in this specific context).  The aim is to provide a spectral description of this operator in the regime $\epsilon\to0$. By doing so, one will revisit recent results \cite{BBKOB22, BKOB23,K23} and provide the reader with an alternative and unifying view point, which might be of interest to tackle the spectral theory of Dirac operators. The main results of this article are Theorems \ref{thm.main} \& \ref{thm.main2} below. They have already been proved in \cite{BKOB23, K23} by considering the square of the Dirac operator and by using the well-known strategy of dimensional reduction for Schrödinger operators, see for instance \cite[Chapter 11]{Raymond17}, inspired by the famous Born-Oppenheimer approximation \cite{BO27}. This strategy of reduction to Schrödinger operators is often used to investigate the spectrum of Dirac operators, see, for instance, \cite{ALTMR19}, \cite[Section 8]{BLTRS23}, \cite{MOBP20} or \cite{LOB23}. 
	In the present article, we want to point out another strategy taking advantage of recent theoretical advances in the semiclassical spectral theory of magnetic Schrödinger operators (see, for instance, \cite{BHR22, FLTRVN22} in 2D or \cite{HR22, AAHR23} in 3D for recent applications). As we explain in Section \ref{sec.13}, the effective operators appearing in \cite{BKOB23,K23} are nothing but Taylor approximations of an already known effective pseudodifferential operator obtained in a general context including (non-necessarily selfadjoint) Schrödinger and Dirac operators, see \cite[Chapitre 3]{Keraval}.

	\subsection{Statement of the main results}
	We start with the result concerning a waveguide, that is the case where $s \in I = \mathbb{R}$ in \eqref{eqn:Gamma}. We denote by $N$ the number of negative eigenvalues (with multiplicity) of $D^2_s-\frac{\kappa^2}{\pi^2}$. For $\mu\in \R$, we denote by $\nu_1(0,\mu)$, the first positive eigenvalue of the \emph{transverse} Dirac operator $\sigma_2(-i\partial_t) + \mu\sigma_3$ equipped with the domain
\[\{\psi\in H^1((-1,1),\mathbb{C}^2) : \mp\sigma_1\psi(\pm 1)=\psi(\pm 1)\}\,.\]	

	
	\begin{theorem}\label{thm.main}
		The spectrum of $\mathscr{D}_\epsilon$ is symmetric with respect to $0$. Moreover, we have, for all $\epsilon>0$,
		\[\mathrm{sp}_{\mathrm{ess}}(\mathscr{D}_\epsilon)=(-\infty,-\epsilon^{-1}\nu_1(0,\epsilon m)]\cup[\epsilon^{-1}\nu_1(0,\epsilon m),+\infty)\,.\]	
In addition, for $\epsilon$ small enough, the positive discrete spectrum contains at least $N$ elements, and for all $j\in\{1,\ldots,N\}$, we have
\[\lambda^+_j(\mathscr{D}_\epsilon)=\epsilon^{-1}\nu_1(0,\epsilon m)+\frac{2\epsilon}{\pi}\lambda_j\left(D^2_s-\frac{\kappa^2}{\pi^2}\right)+o(\epsilon)\,.\]		
	\end{theorem}
	
When we are interested in the tubular neighborhood of a closed curve, that is $s \in I = \mathbb{R} \big/\ell \mathbb{Z}$ in \eqref{eqn:Gamma}, the analogue of Theorem \ref{thm.main} reads as follows.
\begin{theorem}\label{thm.main2} For all $j \in \mathbb{N}$, there holds
\[
	\lambda_j^+(\mathscr{D}_\epsilon) =\epsilon^{-1}\nu_1(0,\epsilon m) +\frac2\pi\epsilon\lambda_j\left(\left(D_s + \frac{\pi+2}\ell\right)^2 - \frac{\kappa^2}{\pi^2}\right) + o(\epsilon).
\]
\end{theorem}

\begin{remark} In Theorem \ref{thm.main2}, using the asymptotic expansion
\[
	\epsilon^{-1}\nu_1(0,\epsilon m) = \frac{\pi}{4\epsilon} + \frac2\pi m +m^2\big(\frac2\pi - \frac{16}{\pi^3}\big)\epsilon+o(\epsilon),
\]
one recovers the main result of \cite{K23}.
\end{remark}

Theorems \ref{thm.main} \& \ref{thm.main2} are the consequences of the same proof. The only difference is that the effective operator obtained for $I = \mathbb{R}$ and the one obtained for $I = \mathbb{R} \big/\ell \mathbb{Z}$ can be reduced either in $D_s^2 - \frac{\kappa^2}4$ in the first case whereas in the second case, the length of the closed curve appears in a {\it magnetic type} extra term (which vanishes as $\ell \to + \infty$).

\subsection{Strategy of the proof}\label{sec.13}
As we said, the main interest of this article is to point out a new strategy to tackle the spectral analysis of Dirac operators. Let us explain its main lines.
In Section \ref{sec.normalform}, we use the canonical unitary transform associated with the change of variable $\Gamma$ and we reduce the analysis to that of $\mathfrak{D}_\epsilon$, a Dirac operator on a strip when $I = \mathbb{R}$ or a rectangle when $I = \mathbb{R}\big/ \ell \mathbb{Z}$ with homogeneous boundary conditions, see Theorem \ref{thm.normal}. We stress out that this theorem has already been proved in \cite[Section 3.1]{BBKOB22}. We recall its proof for the sake of completeness. Section \ref{sec.sstrip} is devoted to the determination of the essential spectrum and to the study of the dispersion curves of the Dirac operator $\mathscr{D}_\epsilon$ when the curve $\gamma(\mathbb{R})$ is a line, see Proposition \ref{prop.straightstrip}. Once again, this elementary proposition has already been proved, see \cite[Section 2]{BBKOB22}. Notice that our proof highlights the underlying supersymmetric structure of the fibered Dirac operators.

The main novelty of this paper appears in Section \ref{sec.4}. There, we see $\mathfrak{D}_\epsilon$ as a pseudodifferential operator with operator-valued symbol, see Proposition \ref{prop.expansiondepsilon}. The principal operator appearing in the expansion is $\mathfrak{D}_0$ is the operator at infinity studied in Section \ref{sec.sstrip}. The spectral analysis of its principal symbol, denoted by $\mathfrak{d}_0$, has also been completed in Section \ref{sec.sstrip}. At this stage, we are in position to apply a version of a microlocal dimensional reduction known as the Grushin method based on the ideas of Sjöstrand (see \cite{SZ07}) and Martinez (see \cite{Martinez07}) and generalized by Keraval in \cite{Keraval}. In essence, the effective operator derivated in \cite[Theorem 4]{BKOB23} to describe the spectrum of $\mathfrak{D}_\epsilon$ is an approximation of the general effective operator (the \enquote{Schur complement}) involved in the Grushin method, whose expression is already known in a general context. Hopefully, this strategy will inspire further developements in the study of Dirac operators.

The core of this method is the object of Section \ref{sec.parametrix} where we construct a \enquote{parametrix}, that is an approximate inverse of an augmented version of $\mathfrak{D}_\epsilon-z$ (where $z$ lies in the region where we want to describe the specrum), see Proposition \ref{prop.parametrix}. This construction is that given in \cite[Chapitre 3]{Keraval} and is simply based on finding an inverse of the principal symbol $\mathfrak{d}_0-z$ (see Lemma \ref{lem.parametrix}) and on an application of the composition formula for pseudodifferential operators, see \cite[Théorème 2.1.12]{Keraval}.	If we look at the remainders in Corollary \ref{cor.ineq}, we understand that we need to prove some a priori estimates of the eigenfunctions to control their $H^2$-norm with respect to the curvilinear abscissa. This is the aim of Section \ref{sec.microloc} where we prove that the eigenfunctions are essentially microlocalized in a bounded region in $\xi$. This allows to control the remainders in Corollary \ref{cor.ineq}. Finally, Section \ref{sec.spectral} is devoted to the spectral consequences of Corollary \ref{cor.ineq}. In fact, this corollary relates the spectrum of $\mathfrak{D}_\epsilon$ to that of a pseudodifferential operator in one dimension. The symbol of this effective operator is given in \eqref{eq.neff}, whose spectrum can be estimated by classical means (the principal symbol has a unique minimum, which is non-degenerate), see Proposition \ref{prop.lambdajNeff}.

\section{Reduction to a normal form}\label{sec.normalform}

\subsection{The normal form}
In order to perform our spectral analysis, it is convenient to put $\mathscr{D}_\epsilon$ in a normal form. By this, we just mean that we should write the action of operator in the coordinates system $(s,t)$ and conjugate the operator with an appropriate weight to remain in a canonical $L^2$-space.

\begin{theorem}\label{thm.normal}
	Let us consider the Dirac operator acting in the Hilbert space $L^2(S_1,\mathrm{d}s\mathrm{d}t)$ as
\[\mathfrak{D}_\epsilon=\epsilon\frac{\sigma_1}{2}\left[(1-\epsilon t\kappa)^{-1}\left(D_s +\frac{\pi}{\ell}\right)+\left(D_s +\frac{\pi}{\ell}\right)(1-\epsilon t\kappa)^{-1}\right]+\sigma_2D_t+m\epsilon\sigma_3\,,\]	
with domain
\[\mathrm{Dom}(\mathfrak{D}_\epsilon)=\{\psi\in H^1(S_1,\mathbb{C}^2) : \psi_2(s,\pm 1)=\mp\psi_1(s,\pm 1)\}\,,\]
where $\ell$ is the length of the curve $\gamma$ ($\ell = +\infty$ is the strip case). 	
Then, $\mathscr{D}_\epsilon$ is unitarily equivalent to $\epsilon^{-1}\mathfrak{D}_\epsilon$. 
Moreover, the spectrum of $\mathfrak{D}_\epsilon$ is symmetric with respect to $0$.
\end{theorem}

\subsection{Proof of Theorem \ref{thm.normal}}

	Let us describe the action of $\mathscr{D}_\epsilon$ in the tubular coordinates $(s,t)$ given by $x=\Gamma(s,t)=\gamma(s)+t\mathbf{n}(s)$ with $\Gamma$ defined in \eqref{eqn:Gamma}. By using the chain rule, we can write ${\rm Jac}\,\Gamma(s,t)^{-T} = \begin{pmatrix}g^{-1}\gamma ' (s)\;, \textbf{n}(s)\end{pmatrix}$, and
	\[
		\nabla_{x} = \gamma'g^{-1}\partial_s + \textbf{n}\partial_t\,,
	\]
	with $g\colon (s,t) \mapsto (1-t\kappa(s))$.
 We obtain that
	 \[
	 \mathscr{D}_\epsilon = \sigma\cdot D + m\sigma_3
	 	= g^{-1}(\sigma\cdot\gamma')(D_s)
		+(\sigma\cdot\textbf{n})(D_t)
		+m\sigma_3
	 \]
	 which acts on the weighted space $L^2(S_\epsilon,g\mathrm{d}s\mathrm{d}t)$. We flatten the metric thanks to a change of unknown $\tilde u = g^{1/2} u$. More precisely,  by conjugation by $g^{1/2}$, we obtain 
	  \[\begin{split}
		-ig^{1/2}\partial_sg^{-1/2}= -i\partial_s  - \frac{it\kappa'}{2g}
	 	\,,\quad \text{ and }\quad
		-ig^{1/2}\partial_tg^{-1/2}   = -i\partial_t  - \frac{i\kappa}{2g}\,,
	 \end{split}\]
	 and the unitarily equivalent operator 
	 \[\begin{split}
		&g^{-1}(\sigma\cdot\gamma')(-ig^{1/2}\partial_sg^{-1/2})
		+(\sigma\cdot\textbf{n})(-ig^{1/2}\partial_tg^{-1/2} )
		+m\sigma_3\,,
		\\&=
		g^{-1}(\sigma\cdot\gamma')(-i\partial_s - \frac{it\kappa'}{2g})
		+(\sigma\cdot\textbf{n})(-i\partial_t - \frac{i\kappa}{2g})
		+m\sigma_3\,.
	 \end{split}\] 	
	Let us recall that the Dirac equation is covariant. In particular, if $\gamma'(s) = e^{i\theta}$, then $\partial_s\theta(s) = \kappa(s)$, 
	\[
		e^{i\sigma_3\theta/2}\sigma\cdot \gamma'e^{-i\sigma_3\theta/2} = \sigma_1\,,
		\quad
		e^{i\sigma_3\theta/2}\sigma\cdot \textbf{n}e^{-i\sigma_3\theta/2} = \sigma_2\,,
	\]
	and
	\[
		e^{i\sigma_3\theta/2}(-i\partial_s)e^{-i\sigma_3\theta/2}
		=
		(-i\partial_s) 
		-
		\frac{\kappa \sigma_3}{2}\,.
	\]
	Nevertheless, when $I = \mathbb{R}\big/ \ell \mathbb{Z}$, $\left(s\mapsto e^{i\sigma_3\theta(s)/2}\right)$ is not necessarily $\ell$-periodic. Thus, the conjugation by this function does not preserve the Hilbert space $L^2(S_\epsilon,\C^2)$ which by definition encodes $\ell$-periodic conditions in the $s$-variable. An alternative is to consider $K_0\colon I\to \R$, the function defined for $s\in I$ by 
	\[
		K_0(s) = \int_0^s\left(\frac{\kappa}{2} - \frac{\pi}{\ell}\right)\dd s\,.
	\]
	Then, if $\alpha \in \R$ is such that $\gamma'(0)=e^{i\alpha}$, the function  $\left(s\in I\mapsto e^{-i\left(\sigma_3\frac{K_0 + \alpha}{2} + \frac{K_0-\alpha}{2}\right)}\right)$ is $s$-periodic. Moreover, there holds
	\[
	e^{i\left(\sigma_3\frac{K_0 + \alpha}{2} + \frac{K_0-\alpha}{2}\right)}\sigma\cdot \gamma'e^{-i\left(\sigma_3\frac{K_0 + \alpha}{2} + \frac{K_0-\alpha}{2}\right)} = \sigma_1\,,
		\quad
		e^{i\left(\sigma_3\frac{K_0 + \alpha}{2} + \frac{K_0-\alpha}{2}\right)}\sigma\cdot \textbf{n}e^{-i\left(\sigma_3\frac{K_0 + \alpha}{2} + \frac{K_0-\alpha}{2}\right)} = \sigma_2\,,
	\]
	and
	\[\begin{split}
		&e^{i\left(\left(\sigma_3\frac{K_0 + \alpha}{2} + \frac{K_0-\alpha}{2}\right)\right)}(-i\partial_s)e^{-i\left(\sigma_3\frac{K_0 + \alpha}{2} + \frac{K_0-\alpha}{2}\right)}
		=
		(-i\partial_s) 
		-
		\frac{\kappa \sigma_3}{2}
		+
		\frac{\kappa}{2}
		-K_0'
		=
		(-i\partial_s) 
		-
		\frac{\kappa \sigma_3}{2}
		+\frac{\pi}{\ell}
		\,.
	\end{split}\]
	so that a conjugation by $e^{-i\left(\frac{\sigma_3\theta}{2} + \frac{\theta}{2} -K_0\right)}$ leads to
		 \[
			g^{-1}\sigma_1(-i\partial_s  - \frac{it\kappa'}{2g}+\frac{\pi}{\ell})
		+\sigma_2(-i\partial_t )
		+m\sigma_3
	 \]
	 which can be rewritten in an explicitly symmetric form as
	\[
			\frac{\sigma_1}{2}\left(g^{-1}(D_s +\frac{\pi}{\ell}) + (D_s+\frac{\pi}{\ell})g^{-1}\right)
		+\sigma_2(D_t )
		+m\sigma_3\,.
	 \]	 
The conclusion follows by using a dilation.

\section{Dirac operator on the straight strip}\label{sec.sstrip}
The goal of this section is to study the spectrum of $\mathfrak{D}_0$, the operator on the straight strip $S_1 := \mathbb{R} \times (-1, 1)$ defined 
as
\[
    \mathfrak{D}_0 = \epsilon \sigma_1 D_s + \sigma_2 D_t + \mu \sigma_3 \quad \text{with} \quad \mu \in \mathbb{R}\,.
\]
The domain of $\mathfrak{D}_0$ is 
\[ 
    \mathrm{Dom}(\mathfrak{D}_0) = \left\{ \psi \in H^1(S_1, \mathbb{C}^2) : \mp \sigma_1 \psi(\cdot, \pm 1) = \psi(\cdot, \pm 1) \right\}\,.
\]
%
%
Thanks to the Fourier transform, it is sufficient to study the family of one-dimensional Dirac operators defined 
for $\xi \in \mathbb{R}$, as
\[
    \mathfrak{d}_0 = \mathfrak{d}_{0,\xi,\mu} := \xi \sigma_1 + \sigma_2 D_t + \mu \sigma_3\,,
\]
acting on 
\[ 
    \mathrm{Dom}(\mathfrak{d}_{0,\xi,\mu}) = \left\{ \psi \in H^1((-1,1), \mathbb{C}^2) : \mp \sigma_1 \psi(\pm 1) = \psi(\pm 1) \right\}\,.
\]
The properties of this family of operators are gathered in the following proposition whose proof can be found further in Appendix \ref{sec.proofSS}.
\begin{proposition}\label{prop.straightstrip}
	The following holds.
\begin{enumerate}[label=\arabic*.]
    \item \label{pt.1xi0} For all $\xi,\mu\in\R$, the operator $\mathfrak{d}_{0,\xi,\mu}$ is self-adjoint with compact resolvent and is invertible. Its spectrum is symmetric with respect to $0$ and the eigenvalues are simple and denoted by
\[
    {\rm Spec}(\mathfrak{d}_{0,\xi,\mu}) = \{\pm \nu_j(\xi,\mu) \mid j \geq 1\}\,,
\]
where  $(\nu_j(\xi,\mu))_{j\geq 1}\subset (0,+\infty)$ is an increasing sequence that tends to $+\infty$. Moreover, for all $j$, $\nu_j(\xi,\mu)$ are smooth functions of $\xi$ and $\mu$.
\item \label{pt.2xi0}For \(\xi = 0\), $\mu\in\R$ and \(j \geq 1\), we have
\[\begin{aligned}
		&\nu_1(0, \mu) \underset{\mu\to-\infty}{\longrightarrow}0\,,
		\\
		&\nu_1(0, \mu) \underset{\mu\to+\infty}{\sim}|\mu|\,,
		\\
		&\nu_j(0, \mu) \underset{\mu\to\pm\infty}{\sim}|\mu|\,,\quad j\geq 2\,,
	\end{aligned}
	\]
	(see Figure \ref{fig:plotDispCurves}) and around $0$,
\[
\nu_1(0,\mu)
=
\frac{\pi}{4} 
+ \frac{2}{\pi}  \mu
+ \left(- \frac{16}{\pi^{3}} + \frac{2}{\pi^{}}\right) \mu^{2} 
+ \left(\frac{256}{\pi^5} - \frac{80}{3 \pi^3}\right)\mu^{3}  
+\left(- \frac{5120}{\pi^{7}} - \frac{8}{\pi^{3}} + \frac{1792}{3 \pi^{5}}\right)\mu^4 
+ \mathscr{O}\left(\mu^{5}\right)
\]	
with an associated normalized eigenvector
\[
\varphi_{0,\mu,1}\colon t\longmapsto c\left(\frac{\cos(kt)}{\cos(k)}e_2 - \frac{\sin(kt)}{\sin(k)}e_1\right)\,,
\]
where \(e_1, e_2\) denote the vectors of the canonical basis of \(\mathbb{R}^2\), $k = k(\mu)$ is the first positive solution to
\[\mu = -\frac{k}{\tan(2k)}\,,\]
and 
\[
	c = \left(\frac{k\sin^2(2k)}{4k-\sin(4k)}\right)^{\frac{1}{2}}.
\]
$\varphi_{0,\mu,-1} := \sigma_1\varphi_{0,\mu,1}$ is a normalized eigenvector associated to $-\nu_1(0, \mu)$.
\item  \label{pt.3xi0}
For $\xi,\mu\in\R$, $j\geq 1$, we have
\[
	\nu_j(\xi, \mu) = \sqrt{\xi^2 + \nu_j(0, \mu)^2}\,,
\]
and 
\[\begin{split}
	\varphi_{\xi,\mu,1} &= c_\xi\left(I_2 +\frac{\xi}{\sqrt{\xi^2+\nu(0,\mu)^2} + \nu(0,\mu)}\sigma_1\right)\varphi_{0,\mu,1}\,,
	\\
	\varphi_{\xi,\mu,-1} &= c_\xi\left(\sigma_1 -\frac{\xi}{\sqrt{\xi^2+\nu(0,\mu)^2} + \nu(0,\mu)}I_2\right)\varphi_{0,\mu,1}\,,
\end{split}\]
are normalized eigenvectors associated with $\pm \nu_1(\xi, \mu)$ where
\[
c_\xi = \left(\frac{\sqrt{\xi^2+\nu_1(0, \mu)^2} + \nu_1(0, \mu)}{2\sqrt{\xi^2+\nu_1(0, \mu)^2}}\right)^{\frac{1}{2}}\,.
\] 
\item We have
\[
	{\rm Spec}(\mathfrak{D}_0) = (-\infty, -\nu_1(0,\mu)]\cup[\nu_1(0,\mu), +\infty)\,,
\]
and
\[\mathrm{sp}_{\mathrm{ess}}(\mathscr{D}_\epsilon)
=
\epsilon^{-1}\mathrm{sp}(\mathfrak{D}_0)
=
(-\infty,-\epsilon^{-1}\nu_1(0,m\epsilon)]\cup[\epsilon^{-1}\nu_1(0,m\epsilon),+\infty)\,.\]	
where $\mathscr{D}_\epsilon$ is defined in \eqref{Def.Diracmain}.

\end{enumerate}
\end{proposition}

Figure \ref{fig:plotDispCurves} plots some dispersion curves $\mu\mapsto \nu_j(0,\mu)$.
\begin{figure}[htbp]
    \centering
    \includegraphics[width=1\textwidth]{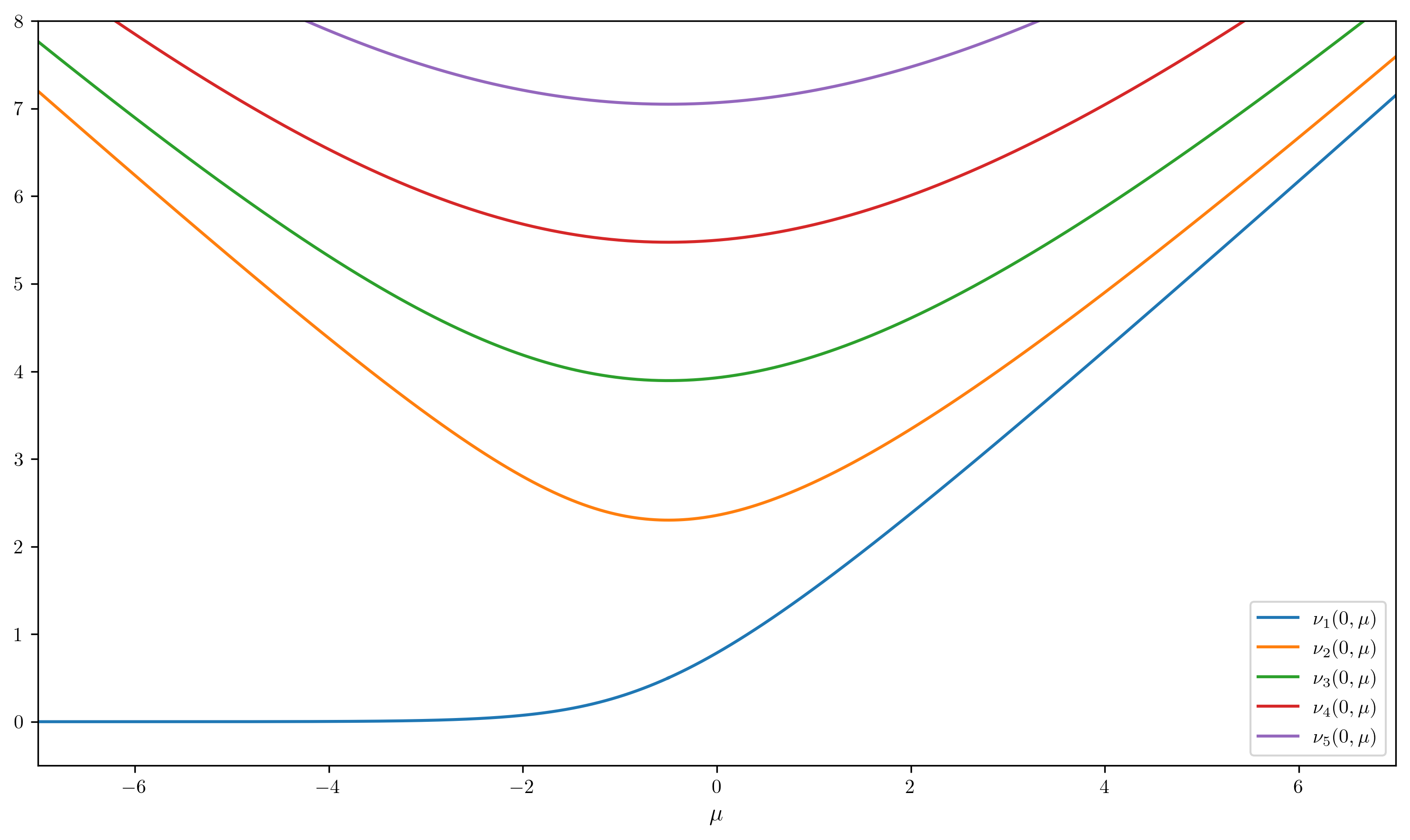}
    \caption{Curves of the eigenvalues \((\nu_j(0,\mu))_{j=1,\dots,5}\)}
    \label{fig:plotDispCurves}
\end{figure}

	\section{Dimensional reduction}\label{sec.4}
Let us now focus on the description of the discrete spectrum. Since the spectrum is symmetric with respect to $0$, we may focus on the possible (necessarily discrete) positive spectrum of $\mathfrak{D}_\epsilon$ in $[0,\nu_1(0,m\epsilon))$.
We recall that
\[\mathfrak{D}_\epsilon=\epsilon\frac{\sigma_1}{2}\left[(1-\epsilon t\kappa)^{-1}\left(D_s +\frac{\pi}{\ell}\right)+\left(D_s +\frac{\pi}{\ell}\right)(1-\epsilon t\kappa)^{-1}\right]+\sigma_2D_t+m\epsilon\sigma_3\,.\]

\subsection{A pseudodifferential view point}
The operator $\mathfrak{D}_\epsilon$ can be regarded as an $\epsilon$-pseudo-differential operator with an operator-valued (Weyl) symbol relative to the curvilinear coordinate $s$. In this article, we employ the Weyl quantization. The formal expression for a symbol  $p$  is given by\begin{equation}\label{eq.pW}
p^W\psi(s)=\frac{1}{2\pi \epsilon}\int_{\R^2}e^{i(s-y)\xi/\epsilon}p\left(\frac{s+y}{2},\xi\right)\psi(y)\mathrm{d}y\mathrm{d}\xi\,.
\end{equation}
For instance, the differential operator $\sigma_1 \epsilon D_s$ is such an operator, and its symbol is simply $\xi \sigma_1$. This symbol belongs to the class of symbols $S^1(\mathbb{R}^2,\mathscr{L}(\mathscr{A},\mathscr{B}))$, where we used the notation for all $\eta \in \left[0, \frac{1}{2}\right)$ (removing the reference to $\eta$ when $\eta = 0$).
\begin{multline*}S^m_\eta(\mathbb{R}^2,\mathscr{L}(\mathscr{A},\mathscr{B}))\\
=\{p\in\mathscr{C}^\infty(\mathbb{R}^2,\mathscr{L}(\mathscr{A},\mathscr{B})) : \forall (\alpha,\beta)\,,\exists C_{\alpha,\beta}>0\,, \|\partial^\alpha_s\partial^\beta_\xi p\|_{\mathscr{L}(\mathscr{A},\mathscr{B})}\leq C_{\alpha,\beta}\epsilon^{-(\alpha+\beta)\eta}\langle \xi\rangle^{m} \}\,,\end{multline*}
with 
\begin{equation}\label{eq.AB}
\mathscr{A}=\{u\in H^1(I,\mathbb{C}^2) : u_2(\pm 1)=\mp u_1(\pm 1)\}\,,\quad \mathscr{B}=L^2(I,\mathbb{C}^2)\,,
\end{equation}
where $\mathscr{A}$ is equipped with the usual $H^1$-norm. When there is no ambiguity we write $S_\eta^m(\R^2) := S^m_\eta(\mathbb{R}^2,\mathscr{L}(\mathscr{A},\mathscr{B}))$ and $S^m(\R^2) := S^m(\mathbb{R}^2,\mathscr{L}(\mathscr{A},\mathscr{B}))$.

\begin{remark}
When $\ell < +\infty$, we additionally assume that the symbols are $\ell$-periodic in the spatial variable and we consider the {\it modified} Weyl quantization
\[
 p^{\widetilde{W}}\psi(s)=\frac{1}{2\pi \epsilon}\int_{\R^2}e^{i(s-y)(\xi-\frac{\pi}\ell)/\epsilon}p\left(\frac{s+y}{2},\xi\right)\psi(y)\mathrm{d}y\mathrm{d}\xi\,.
\]
which quantizes the symbol $\big((s,\xi)\mapsto\xi\big)$ in $(\xi)^{\widetilde{W}}= \epsilon(D_s + \frac\pi\ell)$. We refer to \cite[\S 5.3.1]{Zworski} for a short discussion on the quantization of periodic symbols and, from now on and for the sake of presentation, we focus on the strip case ($\ell = +\infty$). 
\end{remark}

\subsubsection{Standard tools about pseudodifferential operators}

For further uses, we recall several results about $\epsilon$-pseudodifferential operators with operator valued symbols. They can be found in \cite{Keraval}. The first one concerns the composition theorem that can be found, {\it e.g.}, in \cite[Théorème 2.1.12]{Keraval}.
\begin{theorem}[Composition]\label{thm:comp_thm} Let $\mathscr{X}, \mathscr{Y}, \mathscr{Z}$ be three Banach spaces. Let $a \in S_\eta^m(\R^2,\mathscr{L}(\mathscr{Y},\mathscr{Z}))$ and $b \in S_\eta^d(\R^2,\mathscr{L}(\mathscr{X},\mathscr{Y}))$ there holds
\[
	a^W b^W = (a\sharp b)^W
\]
with $a\sharp b \in S_\eta^{m+d}(\R^{2},\mathscr{L}(\mathscr{X},\mathscr{Z}))$ and verifying
\[
	a\sharp b = \sum_{k=0}^N\frac1{k!}\left(\frac{i\epsilon}2\right)^k\sigma(D_X;D_Y)^k(a(X)b(Y))|_{Y=X} + \epsilon^{(N+1)(1-2\eta)}r_{N+1},
\]
with $r_{N+1} \in S_\eta^{m+d}(\R^{2},\mathscr{L}(\mathscr{X},\mathscr{Z}))$. Here $X = (s,\xi), Y= (\varsigma,p)$ and $\sigma(D_X;D_Y) = -\partial_\xi\partial_\varsigma+\partial_s\partial_p$.

In particular, for all $N \in \N$, if $\rm{supp}(a) \cap \rm{supp}(b) = \emptyset$, $a\sharp b = \epsilon^{(N+1)(1-2\eta)}r_{N+1}$.
\end{theorem}
We will also use the following Calder\'on-Vaillancourt theorem adapted for operator-valued pseudodifferential operators (see \cite[Théorème 2.1.16]{Keraval}).
\begin{theorem}[Calder\'on-Vaillancourt]\label{thm.CV} Let $\mathscr{X}, \mathscr{Y}$ be Banach spaces. If $a \in S_\eta(\R^2,\mathscr{L}(\mathscr{X},\mathscr{Y}))$ then
\[
a^W: L^2(\R,\mathscr{X})\to L^2(\R,\mathscr{Y})
\]
and there exists $\epsilon_0,C,M > 0$ such that for all $\epsilon\in(0,\epsilon_0)$ there holds
\[
\|a^W\|_{\mathscr{L}(L^2(\R,\mathscr{X}),L^2(\R,\mathscr{Y}))} \leq C \sum_{|\alpha| \leq M}\epsilon^{\frac{|\alpha|}2}\sup_{(s,\xi)\in \R^2}\|\partial^\alpha a(s,\xi)\|_{\mathscr{L}(\mathscr{X},\mathscr{Y})}.
\] 
\end{theorem}

\subsubsection{A pseudodifferential viewpoint}
The symbol of $\mathfrak{D}_\epsilon$ can be easily computed since it is a differential operator of order $1$ (see \cite[Theorem 4.5]{Zworski}) and we have
\[\mathfrak{D}_\epsilon=\mathfrak{d}^W_\epsilon\,, \mbox{ with }\quad\mathfrak{d}_\epsilon=\sigma_2D_t+(1-t\epsilon\kappa(s))^{-1}\xi \sigma_1+\mu \sigma_3\,.\]
Remark that $\mathfrak{d}_\epsilon \in S^1(\R^2,\mathscr{L}(\mathscr{A},\mathscr{B}))$ because for $\epsilon$ small enough $(1-t\epsilon\kappa(s))^{-1}$ belongs to $S^0(\mathbb{R}^2)$ and by the composition Theorem \ref{thm:comp_thm}$, (1-\epsilon t\kappa(s))^{-1} \sigma_1\epsilon \left(D_s +\frac{\pi}{\ell}\right)$ is an $\epsilon$-pseudodifferential operator whose symbol belongs to $S^1(\mathbb{R}^2)$. Actually, by observing that
\[(1-\epsilon t\kappa(s))^{-1}=\sum_{j=0}^{+\infty} \kappa(s)^jt^j\epsilon^j\,,\]
we have the following.

\begin{proposition}\label{prop.expansiondepsilon}
We can write
\[\mathfrak{D}_\epsilon=\mathfrak{D}_0+\epsilon\mathfrak{D}_1+\ldots+\epsilon^N\mathfrak{R}_N(\epsilon)\,,\]
where
\[\mathfrak{D}_0=\epsilon\sigma_1 ( D_s + \frac\pi\ell)+\sigma_2 D_t+\mu\sigma_3\,,\quad \mathfrak{D}_1=\frac{t}{2}\sigma_1\epsilon(\kappa ( D_s+ \frac\pi\ell)+( D_s+ \frac\pi\ell)\kappa)\,,\]
and $\mathfrak{R}_N(\epsilon)$ being an $\epsilon$-pseudodifferential operator with operator-valued symbol in $S^1(\mathbb{R}^2)$.	

Moreover, we can also write the following expansion (in the topology of $S^1(\mathbb{R}^2,\mathscr{L}(\mathscr{A},\mathscr{B}))$) of the symbol $\mathfrak{d}_\epsilon$ of $\mathfrak{D}_\epsilon$:
\[\mathfrak{d}_\epsilon=\mathfrak{d}_0+\epsilon\mathfrak{d}_1+\ldots+\epsilon^N\mathfrak{r}_N(\epsilon)\,,\]
where
\begin{equation}\mathfrak{d}_0=\xi\sigma_1 +\sigma_2 D_t+\mu\sigma_3\,,\quad \mathfrak{d}_1=t\kappa(s)\xi\sigma_1\,,\quad \mathfrak{d}_2=(t\kappa(s))^2\xi\sigma_1\,.
\label{eqn:defdbarj}
\end{equation}
\end{proposition}

\subsection{Microlocalization estimates}\label{sec.microloc}
This section is devoted to the proofs of Lemmas \ref{lem.xiborne} to \ref{lem.xiborne0}. Lemma \ref{lem.xiborne} establishes elliptic estimates. Lemma \ref{lem.xiborne0} further states that these eigenfunctions are microlocalized near $\xi = 0$. These considerations will be instrumental in estimating the remainders in our approximate parametrix construction in Section \ref{sec.parametrix}. The reader may skip the proofs on a first reading.

\begin{lemma}[Elliptic estimates]\label{lem.xiborne}
	Let $M>0$ and $k\in \N^*$. There exist $\epsilon_0>0$ and $C > 0$ such that for all $\epsilon\in(0,\epsilon_0)$ and all eigenfunctions $\psi$ of $\mathfrak{D}_\epsilon$ associated with an eigenvalue $0\leq \lambda\leq M$ we have	
	\[
    		\|(\langle \xi \rangle^k)^W \psi\|_{L^2(\mathbb{R}, \mathscr{A})} \leq C \|\psi\|\,.
   	 \]
where the Weyl quantization is defined in \eqref{eq.pW}. 
\end{lemma}
\begin{proof}
	Let $M > 0$, $\lambda \in [0,M]$ an eigenvalue of the operator $\mathfrak{D}_\epsilon$ and $\psi$ an associated eigenfunction. 
	As there holds $\mathfrak{D}_\epsilon \psi = \lambda \psi$, it gives 
	\begin{equation}
		\mathfrak{D}_0 \psi = \lambda\psi -\epsilon \left(\xi\sigma_1r_1^\epsilon\right)^W\psi\,,
		\label{eqn:quasivp}
	\end{equation}
	with $r_1^\epsilon\colon(s,t) \mapsto \frac{t\kappa}{1-\epsilon t \kappa}$.
	The proof is divided into several steps.
	\begin{enumerate}[label=\it Step \arabic*.,  leftmargin=0cm, itemindent = 1.5cm]
	\item Let us start with some estimates on \(\epsilon\left(\xi\sigma_1r_1^\epsilon\right)^W\psi\).
	Since $\xi\sigma_1r_1^\epsilon\in S_0^1(\R^2,\mathscr{L}(\mathscr{B}))$ and $\langle\xi\rangle^s\in S_0^s(\R^2,\mathscr{L}(\mathscr{B}))$ for all $s$, Theorem \ref{thm:comp_thm} and the Calderón-Vaillancourt theorem (Theorem \ref{thm.CV}) ensure that there exists a positive constant $C$ such that
	\begin{equation}\label{eq.35}
		\|\epsilon(\langle\xi\rangle^k)^W \left(\xi\sigma_1r_1^\epsilon\right)^W\psi\| 
		= 
		\|\epsilon(\langle\xi\rangle^k)^W \left(\xi\sigma_1r_1^\epsilon\right)^W(\langle\xi\rangle^{-(k+1)})^W(\langle\xi\rangle^{(k+1)})^W\psi\|
		\leq C\epsilon\|(\langle\xi\rangle^{(k+1)})^W\psi\|\,.
	\end{equation}
%
%
	\item
	
	Let $f \in L^2(S_1, \mathbb{C}^2)$ be such that $(\langle \xi \rangle^k)^W f \in L^2(S_1, \mathbb{C}^2)$. Define $\psi = \mathfrak{D}_0^{-1} f \in \mathrm{Dom}(\mathfrak{D}_0)$, where $\mathrm{Dom}(\mathfrak{D}_0)$ is as specified in Theorem \ref{thm.normal} (Proposition \ref{prop.straightstrip} ensures that $\mathfrak{D}_0$ is invertible). We assert that $(\langle \xi \rangle^k)^W \psi \in \mathrm{Dom}(\mathfrak{D}_0)$. Indeed, for any $\varphi \in \mathrm{Dom}(\mathfrak{D}_0) \cap \mathscr{C}^\infty_c(\overline{S_1})$ and almost any $\xi \in \mathbb{R}$, since $\mathfrak{d}_0$ is self-adjoint, we have
\[
    \langle \mathfrak{d}_0 \widehat{\varphi}, \widehat{\psi} \rangle 
    =
    \langle \widehat{\varphi}, \mathfrak{d}_0 \widehat{\psi} \rangle 
    =
    \langle \widehat{\varphi}, \widehat{f} \rangle.
\]
Therefore, since $(\langle \xi \rangle^k)^W f \in L^2(S_1, \mathbb{C}^2)$, we obtain the following estimate for the distributional pairing:
\[
\begin{split}
    |\langle (\langle \xi \rangle^k)^W \mathfrak{D}_0 \varphi, \psi \rangle_{\mathcal{D}, \mathcal{D}'}| 
    &= 
    \left| 
    \int_{\mathbb{R}} \langle \xi \rangle^k \langle \mathfrak{d}_0 \widehat{\varphi}, \widehat{\psi} \rangle \, d\xi 
    \right| 
    = 
    \left| 
    \int_{\mathbb{R}} \langle \widehat{\varphi}, \langle \xi \rangle^k \widehat{f} \rangle \, d\xi 
    \right| 
    \\&\leq
    \|\varphi\| \|(\langle \xi \rangle^k)^W f\|,
\end{split}
\]
thus, $(\langle \xi \rangle^k)^W \psi \in \mathrm{Dom}(\mathfrak{D}_0^*) = \mathrm{Dom}(\mathfrak{D}_0)$ and $\mathfrak{D}_0 (\langle \xi \rangle^k)^W \psi = (\langle \xi \rangle^k)^W f$.

	\item  With $\psi$ being a solution to \eqref{eqn:quasivp}, let $f = \lambda \psi - \epsilon (\xi \sigma_1 r_1^\epsilon)^W \psi$. By \eqref{eq.35} and the triangle inequality, we obtain
\begin{equation}\label{eq.47}
\begin{split}
    \|\mathfrak{D}_0 (\langle \xi \rangle^k)^W \psi\|
    &=
    \left\|
        (\langle \xi \rangle^k)^W
        \left(
            \lambda \psi - \epsilon (\xi \sigma_1 r_1^\epsilon)^W \psi
        \right)
    \right\|
    \\
    &\leq 
    M \|(\langle \xi \rangle^k)^W \psi\|
    +
    \epsilon
    \left\|
        (\langle \xi \rangle^k)^W
        (\xi \sigma_1 r_1^\epsilon)^W \psi
    \right\|
    \\
    &\leq 
    M \|(\langle \xi \rangle^k)^W \psi\|
    +
    C \epsilon \|(\langle \xi \rangle^{(k+1)})^W \psi\|
    \,.
\end{split}
\end{equation}
	\item We recall now that for $\psi \in \mathrm{Dom}(\mathfrak{D}_0)$, a standard argument by integration by parts ensures that
\[
    \|\mathfrak{D}_0 \psi\|^2 = \|\epsilon (D_s + \frac\pi\ell) \psi\|^2 + \|(-i \sigma_2 \partial_t + \mu \sigma_3) \psi\|^2\,.
\]
Since $(-i \sigma_2 \partial_t + \mu \sigma_3)$ is invertible and closed, there exists $c > 0$ such that
\[
    \|(-i \sigma_2 \partial_t + \mu \sigma_3) \psi\|^2
    \geq 
    c \left(
        \|\partial_t \psi\|^2 + \|\psi\|^2
    \right)\,.
\]
Therefore, up to taking a different constant $c$, we obtain for all $\psi \in \mathrm{Dom}(\mathfrak{D}_0)$
\begin{equation}\label{eq.SchroLich}
    \|\mathfrak{D}_0 \psi\|
    \geq c \left(
        \|(\langle \xi \rangle)^W \psi\| + \|\psi\|_{L^2(\mathbb{R}, \mathscr{A})}
    \right).
\end{equation}
	\item Combining \eqref{eq.47} and \eqref{eq.SchroLich}, we obtain
\[
    c \left(
        \|(\langle \xi \rangle^{(k+1)})^W \psi\| + \|(\langle \xi \rangle^k)^W \psi\|_{L^2(\mathbb{R}, \mathscr{A})}
    \right)
    \leq
    M \|(\langle \xi \rangle^k)^W \psi\|
    +
    C \epsilon \|(\langle \xi \rangle^{(k+1)})^W \psi\|
\]
so that for $\epsilon$ small enough, there exists $C > 0$ such that
\[
    \|(\langle \xi \rangle^{(k+1)})^W \psi\| + \|(\langle \xi \rangle^k)^W \psi\|_{L^2(\mathbb{R}, \mathscr{A})}
    \leq
    C \|(\langle \xi \rangle^k)^W \psi\|.
\]
We obtain the result by induction on $k$.

\end{enumerate}
\end{proof}

Now, let us turn to the result of microlocalization, which follows from an argument rather similar to that in the article \cite[Section 3.2]{FLTRVN22} devoted to magnetic Schrödinger operators.
\begin{lemma}[Microlocalisation]\label{lem.xiborne0}
Let $K > 0$, $N \in \mathbb{N}$ and $\eta \in \left(0, \frac{1}{2}\right)$. Consider $\Xi_0$, a smooth bounded function that equals $0$ near $0$, and let $\Xi \colon \xi \mapsto \Xi_0(\epsilon^{-\eta}\xi)$.
There exists $\epsilon_0 > 0$ such that for all $\epsilon \in (0, \epsilon_0)$ and all eigenfunctions $\psi$ of $\mathfrak{D}_\epsilon$ associated with an eigenvalue $0 \leq \lambda \leq \nu_1(0, \mu) + K \epsilon^2$, we have
\begin{equation}\label{eq.microloc}
\|\Xi^W \psi\| = \mathscr{O}(\epsilon^N) \|\psi\|\,,
\end{equation}
and for all $k \in \mathbb{N}$,
\begin{equation}\label{eq.controlH2}
\|(\langle \xi \rangle^{k})^W \Xi^W \psi\| = \mathscr{O}(\epsilon^N) \|\psi\|\,,
\end{equation}
where the Weyl quantization is defined in \eqref{eq.pW}. 
\end{lemma}

\begin{proof}
Let us fix $N \in \mathbb{N}$. Without loss of generality, any smooth bounded function that vanishes near $0$ can be replaced by a smooth cutoff function $\Xi_0\colon \R \to [0,1]$ such that for $\xi\in\R$,
\begin{equation}\label{def.Xi}
	\Xi_0(\xi) =
	\begin{cases}
		0 & \text{if } |\xi| \leq R, \\
		1 & \text{if } |\xi| \geq 2R,
	\end{cases}
\end{equation}
for some constant $R > 0$.
We set 
\begin{equation}\label{def.truc00}
	\Xi_\pm = \Xi \mathds{1}_{\pm \xi > 0}\,,
\end{equation}
 and aim to prove that
\[
\|\Xi^W_\pm \psi\| = \mathscr{O}(\epsilon^N) \|\psi\|\,,
\]
which implies the desired result.
The proof is divided into several parts.
\begin{enumerate}[label=\it Step \arabic*.,  leftmargin=0cm, itemindent = 1.5cm]
\item
We have 
\begin{equation}
	\begin{split}
	\big(\mathfrak{D}_0-\lambda\big) \Xi^W_\pm \psi 
	&=
	 \Xi^W_\pm\big(\mathfrak{D}_0-\lambda\big) \psi 
	 =
	  \Xi^W_\pm\big(\mathfrak{D}_\epsilon-\lambda\big) \psi - \epsilon \Xi^W_\pm\left(\xi\sigma_1r_1^\epsilon\right)^W\psi
	  \\&=
	  - \epsilon \Xi^W_\pm\left(\xi\sigma_1r_1^\epsilon\right)^W\psi
	  \,,
	\label{eqn:eqvpxi}
	\end{split}
\end{equation}
with $r_1^\epsilon\colon(s,t) \mapsto \frac{t\kappa}{1-\epsilon t \kappa}$.
Let $\Xi_{1,+}$ be an even smooth function satisfying
\[
	\Xi_{1,+} (\xi)= \left\{	\begin{array}{lcl}
						\xi & \text{if} &   \xi > R\\
						 \frac{R}2 & \text{if} & 0\leq  \xi < \frac{R}2
					\end{array}\right.\,,\quad \Xi_{1,+} \geq \frac{R}2\,,
\]
and define its opposite by $\Xi_{1,-}=-\Xi_{1,+}$ where $R$ appears in \eqref{def.Xi}.
We define
\begin{equation}\label{eq.regu111}
\Xi_{1,\pm,\epsilon}(\xi) = \epsilon^{\eta}\Xi_{1,\pm}(\epsilon^{-\eta}\xi)\,,
\end{equation}
\begin{figure}[htbp]
    \centering
    \includegraphics[width=0.8\textwidth]{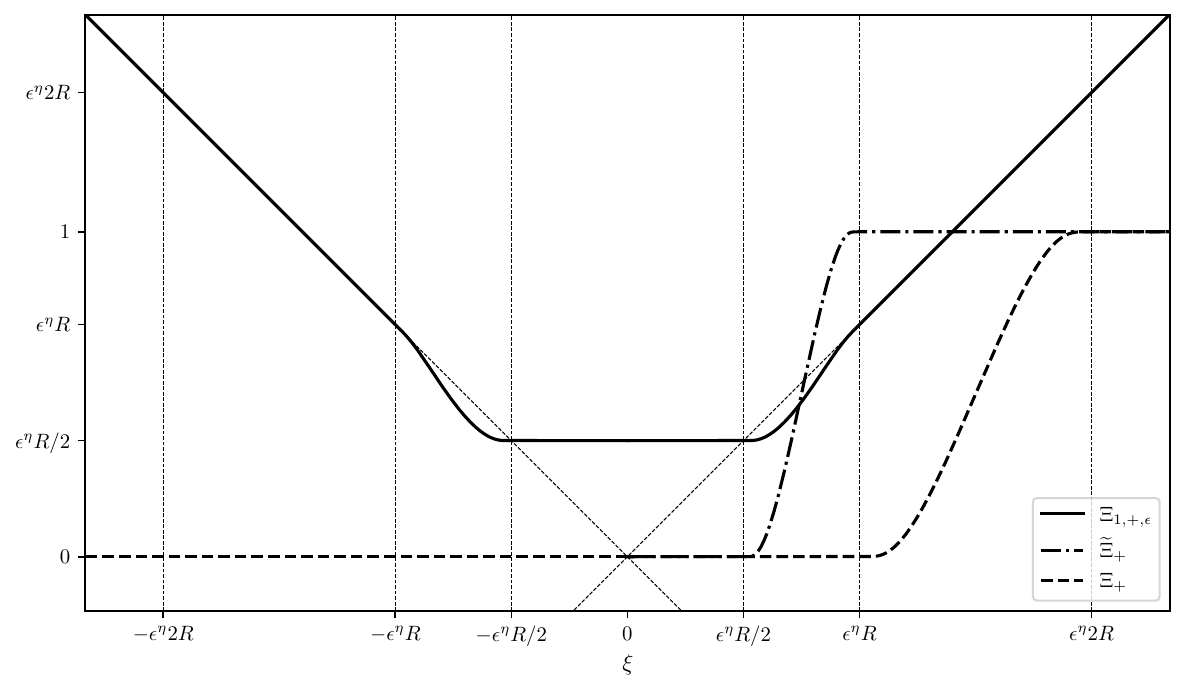}
    \caption{Cutoffs.}
    \label{fig:cutoffs}
\end{figure}
so that $\Xi_{1,\pm,\epsilon} \Xi_\pm = \xi \Xi_\pm$ (see Figure \ref{fig:cutoffs}).
We consider the modified symbol and the associated operator
\begin{equation}\label{eq.d0pm}
	{\mathfrak{d}_0^\pm} := \Xi_{1,\pm,\epsilon}(\xi) \sigma_1 + \sigma_2 D_t + \mu \sigma_3,\quad {\mathfrak{D}_0^\pm} := ({\mathfrak{d}_0^\pm})^W.
\end{equation}
Since the operators are Fourier multipliers, there holds
\[
	\big(\mathfrak{D}_0-\lambda\big) \Xi^W_\pm \psi = \big({\mathfrak{D}_0^\pm}-\lambda\big) \Xi^W_\pm \psi\,,
\]
and \eqref{eqn:eqvpxi} becomes
\begin{equation*}
	\begin{split}
	\big(\mathfrak{D}^\pm_0-\lambda\big) \Xi^W_\pm \psi 
	=
	  - \epsilon \Xi^W_\pm\left(\xi\sigma_1r_1^\epsilon\right)^W\psi
	  \,.
	\end{split}
\end{equation*}
By Lemma \ref{lem.A1}, we obtain
\begin{equation}\label{eq.Xipminverse}
	\Xi_\pm^W \psi = - \epsilon\big({\mathfrak{D}_0^\pm}-\lambda\big)^{-1}\Xi_\pm^W (\sigma_1 \xi r_1^\epsilon)^W\psi\,.
\end{equation}

\item 
Let us focus on the symbol class of the operators involved in \eqref{eq.Xipminverse}. By Lemma \ref{lem.A1}, we have
\begin{equation}\label{eq.classS12}\begin{split}
	&\big({\mathfrak{d}_0^\pm}-\lambda\big)^{-1}\in \epsilon^{-2\eta}S_{2\eta}^{-1}(\R^2,\mathscr{L}(\mathscr{B}))\,,
	\\&
	\Xi_\pm \in S_\eta^0(\R^2,\mathscr{L}(\mathscr{B}))\subset S_{2\eta}^0(\R^2,\mathscr{L}(\mathscr{B}))\,,
	\\&
	(\sigma_1\xi r_1^\epsilon) \in S_0^1(\R^2,\mathscr{L}(\mathscr{B}))\subset S_{2\eta}^1(\R^2,\mathscr{L}(\mathscr{B}))\,.
\end{split}\end{equation} 
Therefore, by the Calderón-Vaillancourt theorem (Theorem \ref{thm.CV}) and \eqref{eq.Xipminverse}, we obtain the estimate
\[
\|\Xi_\pm^W \psi\| \leq C \epsilon^{1-2\eta}\|\psi\|\,,
\]
which is not as strong as the one we expect. To improve this result, it is natural to consider commutators. By \eqref{eq.Xipminverse}, we get
\begin{equation}\label{eq.Xipminverse2}
	\Xi_\pm^W \psi 
	= - \epsilon\big({\mathfrak{D}_0^\pm}-\lambda\big)^{-1} (\sigma_1 \xi r_1^\epsilon)^W\Xi_\pm^W\psi
	-
	\epsilon\big({\mathfrak{D}_0^\pm}-\lambda\big)^{-1}[\Xi_\pm^W, (\sigma_1 \xi r_1^\epsilon)^W]\psi
	\,.
\end{equation}
By the Calderón-Vaillancourt theorem and \eqref{eq.classS12}, we have
\[
	\|- \epsilon\big({\mathfrak{D}_0^\pm}-\lambda\big)^{-1} (\sigma_1 \xi r_1^\epsilon)^W\Xi_\pm^W\psi\|\leq C\epsilon^{1-2\eta}\|\Xi_\pm^W\psi\|\,,
\]
so that \eqref{eq.Xipminverse2} and the triangle inequality imply, for $\epsilon$ small enough:
\begin{equation}\label{eq.Xipminverse3}
	\|\Xi_\pm^W \psi \|
	\leq C\|
	\epsilon\big({\mathfrak{D}_0^\pm}-\lambda\big)^{-1}[\Xi_\pm^W, (\sigma_1 \xi r_1^\epsilon)^W]\psi\|
	\,.
\end{equation}
\item

%
Let us study the commutator that appears in \eqref{eq.Xipminverse3}. By Theorem \ref{thm:comp_thm} with $a := \Xi_\pm$ and $b:= (\sigma_1 \xi r_1^\epsilon)$, the symbol of $[\Xi_\pm^W, (\sigma_1 \xi r_1^\epsilon)^W]$ satisfies for any $M \in \N^*$,
\begin{align*}
	&\Xi_\pm\sharp(\sigma_1 \xi r_1^\epsilon) -(\sigma_1 \xi r_1^\epsilon)\sharp\Xi_\pm  
	\\&=  \sum_{k=1}^{M}\frac{1}{k ! }\left(\frac{i\epsilon}2\right)^k \sigma(D_X;D_Y)^k(a(X)b(Y) - a(Y)b(X))|_{X = Y} + \epsilon^{(M+1)(1-2\eta)}R_\epsilon\\
		& =: \epsilon (\widehat{r}_1) +  \epsilon^{(M+1)(1-2\eta)}R_\epsilon\,,
\end{align*}
with $\widehat{r}_1, R_\epsilon \in S^{1}_{\eta}(\R^2,\mathscr{L}(\mathscr{B}))\subset S^{1}_{2\eta}(\R^2,\mathscr{L}(\mathscr{B}))$. 
The Calderón-Vaillancourt theorem and \eqref{eq.classS12} yield the following estimates 
\begin{equation*}\label{eq.term21}\begin{split}
&\|\epsilon\big({\mathfrak{D}_0^\pm}-\lambda\big)^{-1}\epsilon^{(M+1)(1-2\eta)}R_\epsilon^W\psi\|
\leq C\epsilon^{(M+2)(1-2\eta)}\|\psi\|
\\&
\|\epsilon\big({\mathfrak{D}_0^\pm}-\lambda\big)^{-1}\epsilon (\widehat{r}_1)^W\psi\|\leq C\epsilon^{2(1-\eta)}\|\psi\|\,.
\end{split}\end{equation*}
Therefore, choosing $M$ as $M\geq 2\eta (1-2\eta)^{-1}$ there holds $(M+2)(1-2\eta) \geq 2(1-\eta)$ and for $\epsilon$ small enough, by the triangle inequality and \eqref{eq.Xipminverse3}, we obtain the estimate
\begin{equation}\label{eq.step12}
\begin{split}
\|\Xi_\pm^W \psi \| 
&
\leq
\|\epsilon^2\big({\mathfrak{D}_0^\pm}-\lambda\big)^{-1} (\widehat{r}_1)^W\psi\|
+
C\epsilon^{(M+2)(1-2\eta)}\|\psi\|
\\&
\leq C \epsilon^{2(1-\eta)} \|\psi\|,
\end{split}
\end{equation}
which, again, is not as strong as the one we expect.%
\item
To improve this result, we deepen the study of the term $\epsilon^2\big({\mathfrak{D}_0^\pm}-\lambda\big)^{-1} (\widehat{r}_1)^W\psi$.
For $k \geq 1$, writing $X = (s,\xi)$ and $Y = (\varsigma,p)$, one has
\begin{align*}
	\sigma(D_X;D_Y)^k(a(X)b(Y)) & = (-1)^k(\partial_\xi \partial_\varsigma - \partial_s\partial_p)^k(a(\xi)b(\varsigma,p))\\
	& = (-1)^k\left(\sum_{q=0}^k \begin{pmatrix}k\\q\end{pmatrix} (-1)^q\partial_\xi^{k-q}\partial_\varsigma^{k-q} \partial_s^q\partial_p^q\right)(a(\xi)b(\varsigma,p))\\
	& =(-1)^k\Xi_\pm^{(k)}(\xi)(\partial_s^k b)(\varsigma,p)\,,
	\end{align*}
	and similarily,
\[
	\sigma(D_X;D_Y)^k(a(Y)b(X))  = \Xi_\pm^{(k)}(p)(\partial_s^k b)(s,\xi)\,,
\]	
	where we have used that $\partial_s a = 0$. Therefore, 
	\[
	\sigma(D_X;D_Y)^k(a(X)b(Y) - a(Y)b(X))|_{X = Y} = (\partial_s^k b)(s,\xi)\Xi_\pm^{(k)}(\xi)\left((-1)^k - 1\right)\,,
	\]
	and $\widehat{r}_1$ is supported in the support of $\Xi_\pm'$.

%
%
\item
We introduce $\widetilde{\Xi}_\pm$ which equals $1$ on the support of $\Xi_\pm$ and zero near zero (see Figure \ref{fig:cutoffs}).  Since ${\rm supp}(1- \widetilde{\Xi}_\pm)\cap {\rm supp}(\widehat{r}_1) = \emptyset$, 
\[\begin{split}
	&\widetilde{\Xi}_\pm\in S^0_\eta(\R^2,\mathscr{L}(\mathscr{B}))\subset S^0_{2\eta}(\R^2,\mathscr{L}(\mathscr{B}))\,,
	\\&\widehat{r}_1\in S^{1}_{\eta}(\R^2,\mathscr{L}(\mathscr{B}))\subset S^1_{2\eta}(\R^2,\mathscr{L}(\mathscr{B}))\,,
\end{split}\] 
Theorems \ref{thm:comp_thm}, \ref{thm.CV} and \eqref{eq.classS12} imply that there exists $\widehat{r}_2 \in S^1_{2\eta}(\R^2,\mathscr{L}(\mathscr{B}))$ such that $(\widehat{r}_1)^W (1-\widetilde{\Xi}_\pm)^W = \epsilon^{(M+1)(1-2\eta)}(\widehat{r}_2)^W$ and 
\begin{equation*}
\begin{split}
& \|\epsilon^2\big({\mathfrak{D}_0^\pm}-\lambda\big)^{-1}\widehat{r}_1^W (1- \widetilde{\Xi}_\pm)^W\psi\|
\leq\epsilon^{1+(M+2)(1-2\eta)} \|\psi\|\,,
\\
& \|\epsilon^2\big({\mathfrak{D}_0^\pm}-\lambda\big)^{-1}\widehat{r}_1^W \widetilde{\Xi}_\pm^W\psi\| 
\leq \epsilon^{2(1-\eta)}\| \widetilde{\Xi}_\pm^W\psi\|\,.
\end{split}
\end{equation*}
By  \eqref{eq.step12}, the triangle inequality leads to
\begin{equation}\label{eq.36}
\begin{split}
\|\Xi_\pm^W \psi \| 
&\leq
\|\epsilon^2\big({\mathfrak{D}_0^\pm}-\lambda\big)^{-1} (\widehat{r}_1)^W \widetilde{\Xi}_\pm^W\psi\|
+
C\epsilon^{(M+2)(1-2\eta)}\|\psi\|
\\&
\leq 
\epsilon^{2(1-\eta)}\| \widetilde{\Xi}_\pm^W\psi\|
+
C\epsilon^{(M+2)(1-2\eta)}\|\psi\|\,.
\end{split}\end{equation}

\item We prove by induction on $k \geq 1$ that for any bounded smooth function that vanishes near $0$, there exists a positive constant $C > 0$ such that
\[
\|\Xi_\pm^W \psi \| \leq C \epsilon^{2k(1-\eta)} \|\psi\|.
\]
By \eqref{eq.step12}, this holds for $k = 1$. Assume it is true for some $k \geq 1$. Then, by choosing $M$ in \eqref{eq.36} such that $(M+2)(1-2\eta) \geq 2(k+1)(1-\eta)$, and noting that by the induction hypothesis,
\[
\epsilon^{2(1-\eta)} \| \widetilde{\Xi}_\pm^W \psi \| \leq \epsilon^{2(k+1)(1-\eta)} \|\psi\|,
\]
we obtain the result for $k + 1$. Taking $k$ such that $2(k+1)(1-\eta)>N$, we obtain \eqref{eq.microloc}.
\item
Let $k,N\geq 0$. Replacing in the preceding expressions $(\sigma_1\xi r_1^\epsilon)^W$ by $\langle\xi\rangle^k(\sigma_1\xi r_1^\epsilon)^W\langle\xi\rangle^{-k}$ whose symbol belongs (by the Calder\'on-Vaillancourt theorem) to $S^1_0(\R^2,\mathscr{L}(\mathscr{B}))$, we obtain, up to minor modifications, that there exists a positive constant $C$ such that
\[
\|\langle\xi\rangle^k\Xi_\pm^W \psi \| \leq C \epsilon^{N} \|\langle\xi\rangle^k\psi\|.
\]
Lemma \ref{lem.xiborne} ensures then that \eqref{eq.controlH2} holds.
\end{enumerate}
\end{proof}

\subsection{A parametrix}\label{sec.parametrix}
Let us introduce some notations.
\begin{notation}\label{not.proj1}
Let \(\xi, \mu \in \mathbb{R}\), \(z \in \mathbb{C}\) and $\epsilon>0$.
\begin{enumerate}
\item The operators \(\Pi_{\xi,\mu}\) and \(\Pi_{\xi,\mu}^*\) are defined by
\[
\begin{array}{llll}
	\Pi_{\xi,\mu} \colon & \mathscr{B} & \longrightarrow & \mathbb{C} \\
	& \psi & \longmapsto & \langle \varphi_{\xi,\mu,1}, \psi \rangle_{L^2(I)},
\end{array}
\quad
\text{ and }
\quad
\begin{array}{llll}
	\Pi_{\xi,\mu}^* \colon & \mathbb{C} & \longrightarrow & \mathscr{A} \\
	& c & \longmapsto & c \varphi_{\xi,\mu,1}.
\end{array}
\]
\item The operators \(\mathfrak{P}\) and \(\mathfrak{P}^*\) are defined by \(\mathfrak{P} = (\Pi_{\xi,\mu})^W\) and \(\mathfrak{P}^* = (\Pi_{\xi,\mu}^*)^W\).
 \item The operators \(\mathscr{P}_j, \mathscr{P}_\epsilon \colon \mathscr{A} \times \mathbb{C} \longrightarrow \mathscr{B} \times \mathbb{C}\) are defined, for \(j \geq 0\), by
\[
\mathscr{P}_0 = \begin{pmatrix}
\mathfrak{d}_0 - z & \Pi^*_{\xi,\mu} \\
\Pi_{\xi,\mu} & 0
\end{pmatrix},
\quad
\mathscr{P}_j = \begin{pmatrix}
    \mathfrak{d}_j & 0 \\
    0 & 0
\end{pmatrix},
\]
and
\[
\mathscr{P}_\epsilon = \begin{pmatrix}
\mathfrak{d}_\epsilon - z & \Pi^*_{\xi,\mu} \\
\Pi_{\xi,\mu} & 0
\end{pmatrix}
=  \sum_{j \geq 0} \epsilon^j \mathscr{P}_j.
\]
\item The operators \(\mathfrak{M}_\epsilon\) are defined by
\[
\mathfrak{M}_\epsilon = \left(\mathscr{P}_\epsilon\right)^W=\begin{pmatrix}
\mathfrak{D}_\epsilon - z & \mathfrak{P}^* \\
\mathfrak{P} & 0
\end{pmatrix}.
\]
 \end{enumerate}
The function \(\varphi_{\xi,\mu,1}\) is defined in Proposition \ref{prop.straightstrip}, the operators \(\mathfrak{d}_j\) and \(\mathfrak{D}_\epsilon\) are given in Proposition \ref{prop.expansiondepsilon}, and $\mathscr{A} $and $\mathscr{B}$ are defined in \eqref{eq.AB}.
\end{notation}

The purpose of this section is to construct a parametrix for the operator \(\mathfrak{M}_\epsilon\), \emph{i.e.}, an operator \(\mathscr{Q}_\epsilon\) such that \(\mathscr{Q}_\epsilon = \mathscr{Q}_0 + \epsilon \mathscr{Q}_1 + \epsilon^2 \mathscr{Q}_2\) and
\begin{equation}\label{eq.param}
\mathscr{Q}_\epsilon^W \begin{pmatrix}
\mathfrak{D}_\epsilon - z & \mathfrak{P}^* \\
\mathfrak{P} & 0
\end{pmatrix} = \mathrm{Id} + \epsilon^3 \mathfrak{R}_\epsilon,
\end{equation}
where the remainder \(\mathfrak{R}_\epsilon\) is a pseudo-differential operator whose symbol class is adequate for the problem at aim.

In the following lemma, we gather results concerning the symbol classes of the symbols introduced in Notation \ref{not.proj1}, from Propositions \ref{prop.straightstrip}, \ref{prop.expansiondepsilon}, and Lemma \ref{lem.A3}, as well as their implications for the associated pseudodifferential operators when applying Theorem \ref{thm.CV}.

\begin{lemma}\label{lem.symbolclass}
    We have
    \begin{enumerate}[label = \rm (\roman*)]
    	\item \(\nu_1(\cdot,\mu)\in  S^1(\mathbb{R}^2, \mathscr{L}(\C, \C))\), $\mathfrak{d}_j\in S^1(\R^2,\mathscr{L}(\mathscr{A},\mathscr{B}))$, for $j\geq 0$,
        \item \(\Pi_{\xi,\mu} \in S^0(\mathbb{R}^2, \mathscr{L}(\mathscr{B}, \mathbb{C}))\) and \(\Pi_{\xi,\mu}^* \in S^0(\mathbb{R}^2, \mathscr{L}(\mathbb{C}, \mathscr{A}))\),
        \item \(\mathfrak{P} \in \mathscr{L}(L^2(\mathbb{R}, \mathscr{B}), L^2(\mathbb{R}, \mathbb{C}))\) and \(\mathfrak{P}^* \in \mathscr{L}(L^2(\mathbb{R}, \mathbb{C}), L^2(\mathbb{R}, \mathscr{A}))\),
        \item \(\mathscr{P}_j, \mathscr{P}_\epsilon \in S^1\left(\mathbb{R}^2, \mathscr{L}(\mathscr{A} \times \mathbb{C}, \mathscr{B} \times \mathbb{C})\right)\), for $j\geq 0$.
    \end{enumerate}
\end{lemma}

\subsubsection{Leading order term}
The initial step involves constructing an inverse for the symbol $\mathscr{P}_0$ when $\epsilon = 0$. This can be achieved by utilizing Lemma \ref{lem.A4} and performing a straightforward computation, leading to the following lemma.



\begin{lemma}\label{lem.parametrix}
There exist $\delta, \mu_0 > 0$ such that for all $z \in \left(0, \frac{\pi}{4} + \delta\right)$, all $\mu \in (-\mu_0, \mu_0)$, and all $\xi \in \mathbb{R}$, the operator $\mathscr{P}_0 \colon \mathscr{A} \times \mathbb{C} \to \mathscr{B} \times \mathbb{C}$ is bijective and
\[
\mathscr{Q}_0 =\begin{pmatrix}
Q_0&Q_0^+	\\
Q_0^-&Q_0^\pm	
\end{pmatrix}:= (\mathscr{P}_0(z))^{-1} = \begin{pmatrix}
0 \oplus (\mathfrak{d}_0 - z)_\perp^{-1} & \Pi^*_{\xi, \mu} \\
\Pi_{\xi, \mu} & z - \nu_1(\xi, \mu)
\end{pmatrix},
\]
where
\[
(\mathfrak{d}_0 - z)_\perp := (\mathfrak{d}_0 - z)(\mathrm{Id} - \Pi_{\xi, \mu}^* \Pi_{\xi, \mu})
\]
is invertible when viewed as an endomorphism of $\{\varphi_{\xi, \mu, 1}\}^\perp$ and
\[
\begin{array}{lcllcl}
0 \oplus (\mathfrak{d}_0 - z)_\perp^{-1} \colon
& \mathscr{B} = \mathrm{span}\{\varphi_{\xi, \mu, 1}\} & \oplus \{\varphi_{\xi, \mu, 1}\}^\perp & \longrightarrow & \mathrm{span}\{\varphi_{\xi, \mu, 1}\} & \oplus \{\varphi_{\xi, \mu, 1}\}^\perp \\
& c \varphi_{\xi, \mu, 1} & + \psi & \longmapsto & 0 & + (\mathfrak{d}_0 - z)_\perp^{-1} \psi.
\end{array}
\]
Moreover, the symbol $(\mathscr{P}_0(z))^{-1}$ is an element of $S^{1}\left(\mathbb{R}^2, \mathscr{L}(\mathscr{B} \times \mathbb{C}, \mathscr{A} \times \mathbb{C})\right)$ and
 \begin{enumerate}[label = \rm \alph*)]
    	\item $Q_0 \in S^{-1}(\R^2,\mathscr{L}(\mathscr{B},\mathscr{B}))\cap S^{0}(\R^2,\mathscr{L}(\mathscr{B},\mathscr{A}))$,
        \item \(Q_0^- \in S^0(\mathbb{R}^2, \mathscr{L}(\mathscr{B}, \mathbb{C}))\),
        \item \(Q_0^+ \in S^0(\mathbb{R}^2, \mathscr{L}(\mathbb{C}, \mathscr{A}))\),
        \item \(Q_0^\pm \in S^1\left(\mathbb{R}^2, \mathscr{L}(\C)\right)\).
    \end{enumerate}
\end{lemma}

\subsubsection{Subsequent terms}\label{sec.subterms}
For the second step, we seek an expansion \(\mathscr{Q}_\epsilon = \mathscr{Q}_0 + \epsilon \mathscr{Q}_1 + \epsilon^2 \mathscr{Q}_2\) that satisfies \eqref{eq.param}, where $\mathscr{Q}_0$ is defined in Lemma \ref{lem.parametrix}. Let $\mathscr{Q}_j \in S^{1}\left(\mathbb{R}^2, \mathscr{L}(\mathscr{B} \times \mathbb{C}, \mathscr{A} \times \mathbb{C})\right)$ for \(j \in \{1, 2\}\). Initially, the computations are purely formal (see similar computations in \cite[Section 3.1.3]{Keraval}). By Lemma \ref{lem.parametrix}, the leading order term of 
\[
	\left(\mathscr{Q}_\epsilon\right)^W 
	\begin{pmatrix}
		\mathfrak{D}_\epsilon - z & \mathfrak{P}^* \\
		\mathfrak{P} & 0
	\end{pmatrix}
	- \mathrm{Id}
\]
vanishes.
Note that $\mathscr{Q}_0$ and $\mathscr{P}_0$ are Fourier multipliers so that $(\mathscr{Q}_0)^W(\mathscr{P}_0)^W = (\mathscr{Q}_0\mathscr{P}_0)^W = {\rm Id}$.
By Theorem \ref{thm:comp_thm}, the first-order term is given by
\begin{equation}\label{eq.paramT1}
\begin{split}
    (\mathscr{Q}_1)^W (\mathscr{P}_0)^W &+ (\mathscr{Q}_0)^W (\mathscr{P}_1)^W \\
    &= (\mathscr{Q}_1 \mathscr{P}_0 + \mathscr{Q}_0 \mathscr{P}_1)^W 
    + \frac{\epsilon}{2i} \left( \left\{\mathscr{Q}_1, \mathscr{P}_0\right\} + \left\{\mathscr{Q}_0, \mathscr{P}_1\right\} \right)^W 
    + \epsilon^2 r_2,
\end{split}
\end{equation}
where $r_2$ is a remainder term.
To cancel the lower order term, we set
\[
\mathscr{Q}_1 = -\mathscr{Q}_0 \mathscr{P}_1 \mathscr{Q}_0.
\]
By Theorem \ref{thm:comp_thm} and equation \eqref{eq.paramT1}, the second-order term is
\begin{equation}\label{eq.paramT2}
\begin{split}
    &(\mathscr{Q}_2)^W (\mathscr{P}_0)^W + (\mathscr{Q}_1)^W (\mathscr{P}_1)^W + (\mathscr{Q}_0)^W (\mathscr{P}_2)^W 
    + \frac{1}{2i} \left( \left\{\mathscr{Q}_1, \mathscr{P}_0\right\} + \left\{\mathscr{Q}_0, \mathscr{P}_1\right\} \right)^W \\
    &= \left(
    		\mathscr{Q}_2 \mathscr{P}_0 + \mathscr{Q}_1 \mathscr{P}_1 + \mathscr{Q}_0 \mathscr{P}_2
		+
		\frac{1}{2i} \left( \left\{\mathscr{Q}_1, \mathscr{P}_0\right\} + \left\{\mathscr{Q}_0, \mathscr{P}_1\right\} \right)
	\right)^W 
    + \epsilon r_1,
\end{split}
\end{equation}
where $r_1$ is a remainder term. To cancel the lower order term, we set
\[\begin{split}
\mathscr{Q}_2&=-\left(\mathscr{Q}_1 \mathscr{P}_1 + \mathscr{Q}_0 \mathscr{P}_2
		+
		\frac{1}{2i} \left( \left\{\mathscr{Q}_1, \mathscr{P}_0\right\} + \left\{\mathscr{Q}_0, \mathscr{P}_1\right\} \right)
\right)\mathscr{Q}_0
\\&=
\mathscr{Q}_0 \mathscr{P}_1 \mathscr{Q}_0 \mathscr{P}_1\mathscr{Q}_0 
-
\mathscr{Q}_0 \mathscr{P}_2\mathscr{Q}_0 
+
\frac{1}{2i} \left( \left\{\mathscr{Q}_0 \mathscr{P}_1 \mathscr{Q}_0, \mathscr{P}_0\right\} - \left\{\mathscr{Q}_0, \mathscr{P}_1\right\} \right)
\mathscr{Q}_0\,.
%
\end{split}\]
%
For later use, we gather these elements in the following equation:
\begin{equation}\label{eqn:deftermgrushin}
\mathscr{Q}_0 = \mathscr{P}_0^{-1}, \quad \mathscr{Q}_1 = -\mathscr{Q}_0 \mathscr{P}_1 \mathscr{Q}_0, \quad \mathscr{Q}_2 = -\mathscr{Q}_0 \mathscr{P}_2 \mathscr{Q}_0 + \mathscr{Q}_0 \mathscr{P}_1 \mathscr{Q}_0 \mathscr{P}_1 \mathscr{Q}_0 + \frac{1}{2i} \widetilde{\mathscr{Q}}_2.
\end{equation}
with
\[
\widetilde{\mathscr{Q}}_2 = \left( \{\mathscr{Q}_0 \mathscr{P}_1 \mathscr{Q}_0, \mathscr{P}_0\} - \{\mathscr{Q}_0, \mathscr{P}_1\} \right) \mathscr{Q}_0.
\]
%
%
%
%
%
%

Thanks to a straightforward computation using Lemmas \ref{lem.parametrix}, \ref{lem.A3} and \ref{lem.A4} and the facts that $\mathfrak{d}_j\in S^1(\R^2,\mathscr{L}(\mathscr{B}, \mathscr{B}))$ for $j\geq 1$, we obtain the following lemma.

\begin{lemma}\label{lem.explicitQj}
There exist $\delta, \mu_0 > 0$ such that for all $z \in \left(0, \frac{\pi}{4} + \delta\right)$, all $\mu \in (-\mu_0, \mu_0)$, and all $\xi \in \mathbb{R}$, the following holds : 
\begin{enumerate}[label = \rm (\roman*)]
\item We have
\[\mathscr{Q}_1=\begin{pmatrix}
Q_1&Q_1^+	\\
Q_1^-&Q_1^\pm	
\end{pmatrix}:=-\begin{pmatrix}
Q_0\mathfrak{d}_1Q_0&Q_0\mathfrak{d}_1Q_0^+\\
Q_0^-\mathfrak{d}_1Q_0&Q_0^-\mathfrak{d}_1Q_0^+
\end{pmatrix}\,,\]
\[\mathscr{Q}_2=\begin{pmatrix}
Q_2&Q_2^+	\\
Q_2^-&Q_2^\pm	
\end{pmatrix}:=-\begin{pmatrix}
	Q_0\mathfrak{d}_2Q_0&Q_0\mathfrak{d}_2Q_0^+\\
	Q_0^-\mathfrak{d}_2Q_0&Q_0^-\mathfrak{d}_2Q_0^+
\end{pmatrix}+\begin{pmatrix}
Q_0\mathfrak{d}_1Q_0\mathfrak{d}_1Q_0&Q_0\mathfrak{d}_1Q_0\mathfrak{d}_1Q_0^+\\
Q_0^-\mathfrak{d}_1Q_0\mathfrak{d}_1Q_0&Q_0^-\mathfrak{d}_1Q_0\mathfrak{d}_1Q_0^+
\end{pmatrix}+\frac{1}{2i}\widetilde{\mathscr{Q}}_2\,,\]
and, 
\[\begin{split}
\widetilde{\mathscr{Q}}_2&
=\mathscr{Q}_0(\partial_s\mathscr{P}_1)\partial_\xi\mathscr{Q}_0
-\left(\mathscr{Q}_0(\partial_s\mathscr{P}_1)\partial_\xi\mathscr{Q}_0\right)^*
\\
&=\begin{pmatrix}
Q_0\partial_s\mathfrak{d}_1\partial_\xi Q_0&Q_0\partial_s\mathfrak{d}_1\partial_\xi Q_0^+\\
Q_0^{-}\partial_s\mathfrak{d}_1\partial_\xi Q_0&Q_0^-\partial_s\mathfrak{d}_1\partial_\xi Q_0^+
\end{pmatrix}-\begin{pmatrix}
Q_0\partial_s\mathfrak{d}_1\partial_\xi Q_0&Q_0\partial_s\mathfrak{d}_1\partial_\xi Q_0^+\\
Q_0^{-}\partial_s\mathfrak{d}_1\partial_\xi Q_0&Q_0^-\partial_s\mathfrak{d}_1\partial_\xi Q_0^+
\end{pmatrix}^*\,.\end{split}\]
\item The operator symbols $\mathscr{Q}_1$ and $\mathscr{Q}_2$ are elements of $S^1(\mathbb{R}^2,\mathscr{L}(\mathscr{B}\times\C,\mathscr{A}\times\C))$, we have $\widetilde{\mathscr{Q}}_2 = \xi \overline{\mathscr{Q}}_2$ with $\overline{\mathscr{Q}}_2\in S^0(\mathbb{R}^2,\mathscr{L}(\mathscr{B}\times\C,\mathscr{A}\times\C))$, and for $j \in \{1,2\}$, 

 \begin{enumerate}[label = \rm \alph*)]
    	\item $Q_j \in S^{-1}(\R^2,\mathscr{L}(\mathscr{B},\mathscr{B}))\cap 
	S^{0}(\R^2,\mathscr{L}(\mathscr{B},\mathscr{A}))$,
        \item \(Q_j^- 
        \in S^0(\mathbb{R}^2, \mathscr{L}(\mathscr{B}, \mathbb{C}))
        \),
        \item \(Q_j^+ \in S^1(\R^2,\mathscr{L}(\C,\mathscr{A}))\cap S^0(\R^2,\mathscr{L}(\C,\mathscr{B}))
        \),
        \item \(Q_j^\pm \in S^1\left(\mathbb{R}^2, \mathscr{L}(\C)\right)\).
    \end{enumerate}
\end{enumerate}
\end{lemma}

\subsubsection{Consequences of the parametrix construction}
With Lemma \ref{lem.explicitQj} at hand, we get the  following crucial proposition and corollary.
\begin{proposition}\label{prop.parametrix}
There exist $\delta, \mu_0 > 0$ such that for all $z \in \left(0, \frac{\pi}{4} + \delta\right)$, all $\mu \in (-\mu_0, \mu_0)$, and all $\xi \in \mathbb{R}$, the following holds : 
\begin{enumerate}[label = (\rm \roman*)]
\item \label{pt.1paramId}We have
\[
\mathscr{Q}_\epsilon^W\mathscr{P}_\epsilon^W 
= \mathrm{Id} + \epsilon^3 (\mathscr{R}_\epsilon)^W,
\]
where $\mathscr{P}_\epsilon$ is defined in Notation \ref{not.proj1}, 
$\mathscr{Q}_\epsilon = \mathscr{Q}_0 + \epsilon \mathscr{Q}_1 + \epsilon^2 \mathscr{Q}_2$ with $\mathscr{Q}_0$, $\mathscr{Q}_1$, and $\mathscr{Q}_2$ defined in \eqref{eqn:deftermgrushin}, 
and the remainder $\mathscr{R}_\epsilon$ is an element of $S^2(\mathbb{R}^2, \mathscr{L}(\mathscr{A} \times \mathbb{C}, \mathscr{A} \times \mathbb{C})).$
\item
Let
\begin{equation}\label{eq.neff}
n^{\mathrm{eff}}_\epsilon := \nu_1(\xi, \mu) + \epsilon \kappa(s) \xi \langle t \varphi_{\xi, \mu,1}, \sigma_1 \varphi_{\xi, \mu,1} \rangle,
\end{equation}
and 
\[
\mathscr{Q}_\epsilon =: \begin{pmatrix}
Q & Q^+ \\
Q^- & Q^\pm
\end{pmatrix}.
\]
We have
\begin{equation}\label{eq.Qpm}
Q^{\pm} = z - n^{\mathrm{eff}}_\epsilon + \epsilon^2 \xi r_1(s, \xi) + \epsilon^2 \xi^2 r_2(s, \xi),
\end{equation}
with \(r_1, r_2 \in S^0(\mathbb{R}^2)\).
\end{enumerate}
\end{proposition}
\begin{proof}
By Lemma \ref{lem.symbolclass} and Notation \ref{not.proj1}, we have
\[
\mathscr{P}_\epsilon 
=\mathscr{P}_0+\epsilon\mathscr{P}_1+\epsilon^2\mathscr{P}_2+\epsilon^3R_\epsilon\,,\]
with $R_\epsilon\in S^1(\R^2,\mathscr{L}(\mathscr{A}\times\C,\mathscr{B}\times\C))$.
By the computations performed in Section \ref{sec.subterms}, in particular, equations \eqref{eq.paramT1}, \eqref{eq.paramT2}, \eqref{eqn:deftermgrushin}, Theorem \ref{thm:comp_thm}, and Lemma \ref{lem.explicitQj}, we obtain
%
\[
\mathscr{Q}_\epsilon \sharp \mathscr{P}_\epsilon - {\rm Id}
= \epsilon^3\mathscr{R}_\epsilon,
\]
with $\mathscr{R}_\epsilon\in S^2(\R^2,\mathscr{L}(\mathscr{A}\times\C, \mathscr{A}\times\C))$. 
By using the explicit choices of the $\mathscr{Q}_j$ in \eqref{eqn:deftermgrushin} and the explicit expressions of their coefficients in Lemma \ref{lem.explicitQj}, the conclusion follows.

\end{proof}

\begin{corollary}\label{cor.ineq}
There exist $\delta, \mu_0,C > 0$ such that for all $z \in \left(0, \frac{\pi}{4} + \delta\right)$, and all $\mu \in (-\mu_0, \mu_0)$, the following holds : 
\begin{enumerate}[label = (\rm \roman*)]
\item
For all $\psi\in {\rm dom}(\mathfrak{D}_\epsilon)$ such that $\mathfrak{D}_\epsilon\psi \in {\rm dom}(\mathfrak{D}_\epsilon)$, we have
	\[\begin{split}
	\|\psi\|_{L^2(\R,\mathscr{B})}
	&\leq C\left(
		\|(\mathfrak{D}_\epsilon-z)\psi\|_{L^2(\R,\mathscr{B})}
		+\|\mathfrak{P}\psi\|_{L^2(\R)}
		+\epsilon^3\|(\langle\xi\rangle^2)^W\psi\|_{L^2(\R,\mathscr{A})}
		\right)
	\\
	\|({Q}^\pm)^W\mathfrak{P}\psi\|_{L^2(\R)}
	&\leq C\left(\|(\mathfrak{D}_\epsilon-z)\psi\|_{L^2(\R,\mathscr{B})}+C\epsilon^3\|(\langle\xi\rangle^2)^W\psi\|_{L^2(\R,\mathscr{A})}\right)\,.
	\end{split}\]
\item
For all $f\in H^2(\R,\mathbb{C})$, we have
\begin{equation}\label{eq.rightinverse}
	\|(\mathfrak{D}_\epsilon-z)\mathscr{Q}_+^Wf\|_{L^2(\R,\mathscr{B})}
	\leq C\left(\|({Q}^\pm)^Wf\|_{L^2(\R)}
	+\epsilon^2\|(\xi\mathscr{R})^Wf\|_{L^2(\R,\mathscr{A})}
	+\epsilon^3\|(\langle\xi\rangle^2)^Wf\|_{L^2(\R)}\right)\,,
\end{equation}
with $ \mathscr{R}\in S^0(\R^2, \mathscr{L}(\C, \mathscr{A}))$.
\end{enumerate}
	\end{corollary}
\begin{proof}
The identity of Point \ref{pt.1paramId} of Proposition \ref{prop.parametrix} applied to the vector $(\psi,0)^\mathrm{T}$ gives
\[\begin{split}
	&Q^W(\mathfrak{D}_\epsilon-z)\psi+(Q^+)^W\mathfrak{P}\psi=\psi+\epsilon^3 R^W_\epsilon\psi\,,
	\\
	&(Q^-)^W(\mathfrak{D}_\epsilon-z)\psi+(Q^\pm)^W\mathfrak{P}\psi=\epsilon^3 (R^-_\epsilon)^W\psi\,,
\end{split}\]
where $\mathfrak{P}$ is defined in Notation \ref{not.proj1}, \(\mathfrak{D}_\epsilon\) is defined  in Proposition \ref{prop.expansiondepsilon}, $\mathscr{R}_\epsilon, \mathscr{Q}_\epsilon$ are defined in Proposition \ref{prop.parametrix} and
\[
\mathscr{Q}_\epsilon =: \begin{pmatrix}
Q & Q^+ \\
Q^- & Q^\pm
\end{pmatrix}\,,\quad \text{ and }\quad\mathscr{R}_\epsilon =: \begin{pmatrix}
R_\epsilon & R^+_\epsilon \\
R^-_\epsilon & R^\pm_\epsilon
\end{pmatrix},
\]
$R_\epsilon$ is an element of $S^2(\mathbb{R}^2,\mathscr{L}(\mathscr{A},\mathscr{A}))$ and $R_\epsilon^{-}$ is an element of  $S^2(\mathbb{R}^2,\mathscr{L}(\mathscr{A},\C))$. 
By Lemmas \ref{lem.parametrix}, \ref{lem.explicitQj}, Theorems \ref{thm.CV}, and \ref{thm:comp_thm}, the operators $Q^W\colon \mathscr{B}\to \mathscr{A}$, $Q^+\colon \C\to \mathscr{B}$, $Q^-\colon \mathscr{B}\to\C$, $R_\epsilon\sharp \langle\xi\rangle^{-2}\colon \mathscr{A}\to \mathscr{A}$, and $R^-_\epsilon\sharp\langle\xi\rangle^{-2}\colon \mathscr{A}\to \mathscr{A}$, are bounded. More precisely, there exists $C>0$ such that (by the triangle inequality), 
\begin{align*}
	\|\psi\|_{L^2(\R,\mathscr{B})} &\leq \|Q^W(\mathfrak{D}_\epsilon-z)\psi\|_{L^2(\R,\mathscr{B})} + \|(Q^+)^W\mathfrak{P}\psi\|_{L^2(\R,\mathscr{B})} + \epsilon^3 \|R^W_\epsilon\psi\|_{L^2(\R,\mathscr{B})}\\
	& \leq C \|Q^W(\mathfrak{D}_\epsilon-z)\psi\|_{L^2(\R,\mathscr{A})} + C \|\mathfrak{P}\psi\|_{L^2(\R)} + \epsilon^3 \|R^W_\epsilon \psi\|_{L^2(\R,\mathscr{A})}\\
	&= C \|(\mathfrak{D}_\epsilon-z)\psi\|_{L^2(\R,\mathscr{B})} + C \|\mathfrak{P}\psi\|_{L^2(\R)} + \epsilon^3 \|R^W_\epsilon(\langle\xi\rangle^{-2})^W(\langle\xi\rangle^{2})^W \psi\|_{L^2(\R,\mathscr{A})}\\
	& \leq C \|(\mathfrak{D}_\epsilon-z)\psi\|_{L^2(\R,\mathscr{B})} + C \|\mathfrak{P}\psi\|_{L^2(\R)} + C \epsilon^3 \|(\langle\xi\rangle^{2})^W \psi\|_{L^2(\R,\mathscr{A})},
\end{align*}
and,
\[
	\|(Q^\pm)^W(z) \mathfrak{P}\psi\|_{L^2(\R)} \leq C\left(
		\|(\mathfrak{D}_\epsilon - z)\psi\|_{L^2(\R,\mathscr{B})} + \epsilon^3 \|(\langle\xi\rangle^2)^W \psi\|_{L^2(\R,\mathscr{A})}
		\right).
\]
%
The estimate \eqref{eq.rightinverse} follows from similar considerations by taking the adjoint in the operator identity in Proposition \ref{prop.parametrix} applied to the vector $(0,f)$, with the only difference being the sign in front of the Poisson brackets in \eqref{eqn:deftermgrushin}. More precisely, we have by Lemmas \ref{lem.parametrix}, \ref{lem.explicitQj}, Theorems \ref{thm.CV}, and \ref{thm:comp_thm}, and the computations of Section \ref{sec.subterms},
\[
\begin{split}
	\|(\mathfrak{D}_\epsilon-z)\mathscr{Q}_+^Wf\|_{L^2(\R,\mathscr{B})}
	\leq \|\mathfrak{P}({Q}^\pm)^Wf\|_{L^2(\R)}+C\epsilon^3\|(\langle\xi\rangle^2)^Wf\|_{L^2(\R)} + \epsilon^2\|(\widetilde{\mathscr{Q}}_2\mathscr{P}_0)^W\begin{pmatrix}0\\f\end{pmatrix}\|\,,
\end{split}
\]
and
\[
\widetilde{\mathscr{Q}}_2\mathscr{P}_0\begin{pmatrix}0\\1\end{pmatrix}
=
\widetilde{\mathscr{Q}}_2\begin{pmatrix}\mathfrak{P}^*\\0\end{pmatrix}
=\xi \overline{\mathscr{Q}}_2\begin{pmatrix}\mathfrak{P}^*\\0\end{pmatrix}
=\xi \mathscr{R}\,,
\]
with $ \mathscr{R}\in S^0(\R^2, \mathscr{L}(\C, \mathscr{A}))$ and the conclusion follows.
\end{proof}

\subsection{Spectral analysis}\label{sec.spectral}

\subsubsection{Towards the spectrum of the effective operator} 
Let $\mathscr{N}^{\mathrm{eff}}_\epsilon := (n_{\epsilon}^{\mathrm{eff}})^W$. 
The low-lying spectrum of $\mathscr{N}^{\mathrm{eff}}_\epsilon$, whose symbol is given in \eqref{eq.neff}, is easy to estimate by using classical results about pseudodifferential operators in one dimension.
\begin{proposition}\label{prop.lambdajNeff}
Let $N \geq 1$. Assume that the min-max level $\lambda_N$ is negative, where for $j \geq 1$,
\[
\lambda_j = 
\begin{cases}
    \lambda_j\left(D_s^2 - \frac{\kappa^2}{\pi^2}\right)
    & \text{if the curve is unbounded: } I = \mathbb{R}, \\
    \lambda_j\left(\left(D_s + \frac{2-\pi}{\ell}\right)^2 - \frac{\kappa^2}{\pi^2}\right)
    & \text{if the curve is a closed curve: } I = \mathbb{R} / \ell \mathbb{Z}.
\end{cases}
\]
Then, for all $j\in\{1,\dots,N\}$, we have
\begin{equation}\label{eq.lambdajNeff}
\lambda_j(\mathscr{N}^{\mathrm{eff}}_\epsilon)=\nu_1(0,m\epsilon)+\epsilon^2\frac{2\lambda_j}{\pi}+o(\epsilon^2)\,.
\end{equation}
\end{proposition}
\begin{proof}
The proof is divided into several steps.
\begin{enumerate}[label=\it Step \arabic*.,  leftmargin=0cm, itemindent = 1.5cm]
\item 
The principal symbol $\nu_1(\xi,\mu)=\sqrt{\nu_1(0,m\epsilon)^2+\xi^2}$ has a unique minimum at $\xi=0$, which is non-degenerate. This suggests expanding the operator near $\xi=0$. By Lemma \ref{lem.moment}, Proposition \ref{prop.parametrix} and Proposition \ref{prop.straightstrip}, we have that
\[\begin{split}
	n^{\mathrm{eff}}_\epsilon(s,\xi)
	&=\sqrt{\nu_1(0,m\epsilon)^2+\xi^2}+\epsilon\frac{4\kappa(s)\xi}{\pi^2}+\mathscr{O}(\epsilon^2|\xi|+\epsilon\xi^3)\,,
	\\&=
	\nu_1(0,m\epsilon) +\frac{2\xi^2}{\pi} +\epsilon\frac{4\kappa(s)\xi}{\pi^2}+\mathscr{O}(\epsilon^2|\xi|+\epsilon\xi^3+\xi^4)\,,
	\\&=
	\nu_1(0,m\epsilon) 
	+\frac{2}{\pi}\left(\left(\xi+\frac{\epsilon\kappa}{\pi}\right)^2-\epsilon^2\frac{\kappa^2}{\pi^2}\right)
	+\mathscr{O}(\epsilon^2|\xi|+\epsilon\xi^3+\xi^4)\,.
\end{split}\]
More precisely, we can write
\begin{equation}\label{eq.neffdescription}
n^{\mathrm{eff}}_\epsilon(s,\xi)-\nu_1(0,m\epsilon) =\frac{2}{\pi}\left(\left(\xi+\frac{\epsilon\kappa}{\pi}\right)^2-\epsilon^2\frac{\kappa^2}{\pi^2}\right)
+\epsilon^2\xi r_1(s,\xi)+\epsilon\xi^3r_2(s,\xi)+\xi^4r_3 + \xi^6r_4(s,\xi)\,,
\end{equation}
with $r_j\in S^0(\mathbb{R}^2, \mathscr{L}(\C))$ for $j=1,\dots,4$, and $r_3\in \C$.
Note that 
%
\[
\begin{split}
	&e^{iK_1} \frac{2}{\pi} \left(\left(\xi + \frac{\epsilon \kappa}{\pi}\right)^2 - \epsilon^2 \frac{\kappa^2}{\pi^2}\right)^W e^{-iK_1}
	= \epsilon^2 \frac{2}{\pi} \left(\left(D_s + \frac{\pi+2}{\ell}\right)^2 - \frac{\kappa^2}{\pi^2}\right),
\end{split}
\]
where $K_1 \colon s \mapsto \int_0^s \left(\frac{\kappa(\tau)}{\pi} - \frac{2}{\ell}\right) \, \mathrm{d} \tau$, ($\ell =+\infty$ for the strip).
Therefore, we obtain that the min-max levels satisfy, for $j\geq 1$,
\[
\begin{split}
    &\lambda_j\left(\nu_1(0, m\epsilon) 
    + \frac{2}{\pi}\left(\left(\xi + \frac{\epsilon \kappa}{\pi}\right)^2 - \epsilon^2 \frac{\kappa^2}{\pi^2}\right)\right) \\
    &= \begin{cases}
    \nu_1(0, m\epsilon) + \epsilon^2 \frac{2}{\pi} \lambda_j\left(D_s^2 - \frac{\kappa^2}{\pi^2}\right)
    & \text{if the curve is unbounded: } I = \mathbb{R}, \\
    \nu_1(0, m\epsilon) + \epsilon^2 \frac{2}{\pi} \lambda_j\left(\left(D_s + \frac{\pi+2}{\ell}\right)^2 - \frac{\kappa^2}{\pi^2}\right)
    & \text{if the curve is closed: } I = \mathbb{R} / \ell \mathbb{Z}.
    \end{cases}
\end{split}
\]%
\item \label{step2} Let $\eta\in(0,1]$ 
and $ (f_\epsilon)_\epsilon\subset H^\infty(\R)$ such that for all $k\geq 0$, we have $\|(\xi^k)^Wf_\epsilon\| = \mathscr{O}(\epsilon^{\eta k})$.

Let us prove by induction on $k\geq 0$, that for any $r\in S^0(\R^2, \mathscr{L}(\C))$, we have $\|(\xi^kr)^Wf_\epsilon\| = \mathscr{O}(\epsilon^{\eta k})$. By Theorem \ref{thm.CV}, this is true for $k=0$. Assume that there exists $k\geq 1$, such that for all $n\in\{0,\dots,k-1\}$, and all $r\in S^0(\R^2, \mathscr{L}(\C))$, we have $\|(\xi^nr)^Wf_\epsilon\| = \mathscr{O}(\epsilon^{\eta n})$. By Theorem \ref{thm:comp_thm}, there exists $r_{k+1}\in S^{k}(\R^2, \mathscr{L}(\C))$ such that
\[\begin{split}
	 r\sharp\xi^k &= \sum_{n=0}^{k}\frac1{n!}\left(\frac{-\epsilon}{2i}\right)^n(\partial_\xi^n\xi^k\partial_s^n r)+ \epsilon^{(k+1)}r_{k+1}
	\\&=
	\sum_{n=0}^{k}\binom{k}{n}\left(\frac{-\epsilon}{2i}\right)^n(\xi^{k-n}\partial_s^n r)+ \epsilon^{(k+1)}r_{k+1}\,.
\end{split}\]
Therefore, by induction hypothesis and the triangle inequality, there exists $C>0$ such that
\[\begin{split}
	\|(\xi^kr)^Wf_\epsilon\|
	&\leq \|(r)^W(\xi^k)^Wf_\epsilon\|
	+
	C\sum_{n=1}^{k}\epsilon^n\|(\xi^{k-n}\partial_s^n r)^Wf_\epsilon\|+ \epsilon^{(k+1)}\|r_{k+1}f_\epsilon\|
	\\&
	\leq \|(r)^W(\xi^k)^Wf_\epsilon\|
	+ \epsilon^{(k+1)}\|r_{k+1}f_\epsilon\|
	+\mathscr{O}(\epsilon^{1 + \eta (k-1)}) 
	\,.
\end{split}\]
By Theorems \ref{thm:comp_thm} and  \ref{thm.CV}, up to taking a larger $C>0$, we obtain that
\[
	\|(r)^W(\xi^k)^Wf_\epsilon\|\leq C\|(\xi^k)^Wf_\epsilon\| = \mathscr{O}(\epsilon^{\eta k})\,,
\]
and
\[
\epsilon^{(k+1)}\|r_{k+1}f_\epsilon\|\leq C\epsilon^{(k+1)}\|\langle\xi\rangle^{k}f_\epsilon\|= \mathscr{O}(\epsilon^{k+1})\,.
\]
This concludes the proof by induction.

\item Let us focus on the upper bound. Let $j\in\{1,\dots,N\}$. By elliptic regularity, any normalized eigenfunction $f_{j}\in L^2(\R)$ associated with $\lambda_j$ satisfies $\|f_j\|_{H^\infty(\R)}<+\infty$ so that  for any $n\geq 0$, $\|(\xi^n)^W f_j\|_{L^2(\R)} = \epsilon^n\|(-i\partial_s+\frac\pi\ell)^n f_j\|=\mathscr{O}(\epsilon^n)$. The previous step ensures that
\[
	\|\left(\epsilon^2\xi r_1(s,\xi)+\epsilon\xi^3r_2(s,\xi)+\xi^4r_3 + \xi^6r_4(s,\xi)\right)^Wf\| = \mathscr{O}(\epsilon^3)\,.
\]
We deduce that, for all $j\in\{1,\ldots, N\}$,
\[\left\|\left(\mathscr{N}^{\mathrm{eff}}_\epsilon-\left[\nu_1(0,m\epsilon)+\frac{2\epsilon^2}{\pi}\lambda_j\right]\right)f_j\right\|=\mathscr{O}(\epsilon^3)\,,\]
and the spectral theorem provides us with the upper bound.
\item \label{Step4}
Let us focus on the lower bound. Thanks to the upper bound and Proposition \ref{prop.straightstrip}, there exists $\epsilon_0>0$ and for $\epsilon\in(0,\epsilon_0)$, a family $(f_{j,\epsilon})_{j=1,\dots,N}$ of normalized eigenfunctions of $\mathscr{N}^{\rm eff}_\epsilon$ associated with the eigenvalues $(\lambda_j(\mathscr{N}^{\mathrm{eff}}_\epsilon))_{j=1,\dots,N}$. Let $\eta\in(0,1/2)$. The arguments leading to Lemmas \ref{lem.xiborne} and \ref{lem.xiborne0} ensure that for any $k\geq 0$, there exists $C_k>0$ and $\epsilon_0$ such that for any $\epsilon\in(0,\epsilon_0)$,
	\[
    		\|(\langle \xi \rangle^k)^W f_{j,\epsilon}\|_{L^2(\mathbb{R}, \mathscr{A})} \leq C_k\,,
   	 \]
	 and with the notations of the proof of Lemma \ref{lem.xiborne0}, 
	\[
	\|(\langle \xi \rangle^{k})^W \Xi^W f_{j,\epsilon}\| = \mathscr{O}(\epsilon^k) \,.
	\]
	This identity,  the triangle inequality and the fact that $|(\xi^k)(1-\Xi)|\leq C \epsilon^{\eta k}$ (see Figure \ref{fig:cutoffs}) imply that for any $k\geq 0$,
	\[\begin{split}
		\|(\xi^k)^Wf_{j,\epsilon}\| 
		&\leq \|(\xi^k)^W\Xi^Wf_{j,\epsilon}\| + \|(\xi^k)^W(1-\Xi^W)f_{j,\epsilon}\|
		= \mathscr{O}(\epsilon^k+\eta^{k\eta}) = \mathscr{O}(\eta^{k\eta})\,.
	\end{split}\]
	By the induction of \ref{step2}, we get that for $\eta\in(1/3,1/2)$,
\[
	\|\left(\epsilon^2\xi r_1(s,\xi)+\epsilon\xi^3r_2(s,\xi) + \xi^6r_4(s,\xi)\right)^Wf_{j,\epsilon}\| = \mathscr{O}(\epsilon^{2+\eta} + \epsilon^{1+3\eta} + \epsilon^{6\eta}) = o(\epsilon^2)\,.
\]
Note that this argument implies $\|(\xi^4)^W f_{j,\epsilon}\| = \mathscr{O}(\epsilon^{2-})$, which does not allow us to conclude. Nevertheless, the previous argument applied with $(\xi^2)^W f_{j,\epsilon}$ ensures that $\|(\xi^4)^W f_{j,\epsilon}\| = \mathscr{O}(\epsilon^{2\eta})\|(\xi^2)^W f_{j,\epsilon}\|$.

Let 
\[E_N=\underset{1\leq j\leq N}{\mathrm{span}}\, f_{j,\epsilon},\]
and $\psi\in E_N$. We also have $\|(\xi^4)^W \psi\| = \mathscr{O}(\epsilon^{2\eta})\|(\xi^2)^W \psi\|$, and 
\[
	\left\langle \left(\mathscr{N}^{\mathrm{eff}}_\epsilon-\nu_1(0,\epsilon m)\right)\psi,\psi\right\rangle\leq \lambda_N\left(\mathscr{N}^{\mathrm{eff}}_\epsilon-\nu_1(0,\epsilon m)\right)\|\psi\|^2\,. \]
By \eqref{eq.neffdescription}, we have that there exists $C>0$ such that 
\[\begin{split}
	 &\lambda_N\left(\mathscr{N}^{\mathrm{eff}}_\epsilon-\nu_1(0,\epsilon m)\right)\|\psi\|^2\\
	 &\geq 
	\left\langle \left(\frac{2}{\pi}\left(\left(\xi+\frac{\epsilon\kappa}{\pi}\right)^2-\epsilon^2\frac{\kappa^2}{\pi^2}\right)
+\epsilon^2\xi r_1+\epsilon\xi^3r_2+\xi^4r_3 + \xi^6r_4\right)^W\psi,\psi\right\rangle
	\\&\geq
	\left\langle \left(\frac{2}{\pi}\left(\left(\xi+\frac{\epsilon\kappa}{\pi}\right)^2-\epsilon^2\frac{\kappa^2}{\pi^2}\right)
\right)^W\psi,\psi\right\rangle - o(\epsilon^2)\|\psi\|^2 -C\epsilon^{2+2\eta}\|D_s\psi\|^2
	\\&\geq
	\epsilon^2\left(
		\left\langle 
			\frac{2}{\pi}\left(\left(D_s+\frac{\kappa}{\pi}\right)^2-\frac{\kappa^2}{\pi^2}\right)
\psi,\psi\right\rangle - o(1)\|\psi\|^2 -C\epsilon^{2\eta}\|\left(D_s+\frac{\kappa}{\pi}\right)\psi\|^2\right)
	\\&\geq
	\epsilon^2\left(
	(1-\mathscr{O}(\epsilon^{2\eta}))\left\langle 
		\frac{2}{\pi}\left(\left(D_s+\frac{\kappa}{\pi}\right)^2-\frac{\kappa^2}{\pi^2}\right)
	\psi,\psi\right\rangle - o(1)\|\psi\|^2\right)
	\,. 
\end{split}\]
The min-max principle ensures that
\[
\lambda_N\left(\mathscr{N}^{\mathrm{eff}}_\epsilon-\nu_1(0,\epsilon m)\right)
\geq
\epsilon^2\left((1-\mathscr{O}(\epsilon^{2\eta}))\lambda_N - o(1)\right)
\]
and the conclusion follows.
\end{enumerate}

\end{proof}

\subsubsection{Spectral consequences}

Let $N \geq 1$ an integer such that $\lambda_N$ is negative, where $\lambda_N$ is defined in Proposition \ref{prop.lambdajNeff}.
First, we construct quasimodes for $\mathfrak{D}_\epsilon$ by considering the normalized eigenfunctions $(f_{j,\epsilon})_{1\leq j\leq N}$. Using \eqref{eq.rightinverse} with $z=\lambda_j(\mathscr{N}^{\mathrm{eff}}_\epsilon)$ and \eqref{Step4} of the proof of Proposition \ref{prop.lambdajNeff}. This shows that
\begin{equation*}
	\|(\mathfrak{D}_\epsilon-\lambda_j(\mathscr{N}^{\mathrm{eff}}_\epsilon))\mathscr{Q}_+^Wf\|=o(\epsilon^2)\,.
\end{equation*}
With the spectral theorem and Proposition \ref{prop.lambdajNeff}, this implies that, for all $j\in\{1,\ldots,N\}$,
\[\lambda_j^+(\mathfrak{D}_\epsilon)\leq \lambda_j(\mathscr{N}^{\mathrm{eff}}_\epsilon)+o(\epsilon^2)<\nu_1(0,\epsilon m)\,,\]
which shows the existence of at least $N$ positive eigenvalues (with multiplicity) below the threshold of the positive essential spectrum.

For the lower bound, we use the second inequality in Corollary \ref{cor.ineq}, with $z=\lambda^+_j(\mathfrak{D}_\epsilon)$ and $\psi$ in the corresponding eigenspace. Recalling \eqref{eq.Qpm}, Lemma \ref{lem.xiborne}, this gives
	\[\|({\mathscr{N}}^{\mathrm{eff}}_\epsilon-\lambda^+_j(\mathfrak{D}_\epsilon))\mathfrak{P}\psi\|\leq C\epsilon^{3}\|\psi\|+C\epsilon^2\|(\xi r_1+\xi^2 r_2)^W\mathfrak{P}\psi\|\,.\]
With similar considerations as in the proof of Proposition \ref{prop.lambdajNeff}, we have
\[\|(\xi r_1+\xi^2 r_2)^W\mathfrak{P}\psi\|\leq C\epsilon^{\eta}\|\mathfrak{P}\psi\|\,.\]
With the first inequality of Corollary \ref{cor.ineq} and Lemma \ref{lem.xiborne}, we get $\|\psi\|\leq C\|\mathfrak{P}\psi\|$ and thus
	\[\|({\mathscr{N}}^{\mathrm{eff}}_\epsilon-\lambda^+_j(\mathfrak{D}_\epsilon))\mathfrak{P}\psi\|\leq C\epsilon^{2+\eta}\|\mathfrak{P}\psi\|\,.\]
Since $\mathfrak{P}$ is injective on $\ker(\mathfrak{D}_\epsilon-\lambda^+_j(\mathfrak{D}_\epsilon))$, this last estimate and the spectral theorem imply that, for all $j\in\{1,\ldots,N\}$,
\[\lambda_j\left({\mathscr{N}}^{\mathrm{eff}}_\epsilon\right)\leq\lambda^+_j(\mathfrak{D}_\epsilon)+o(\epsilon^2)\,.\]
Theorem \ref{thm.main} follows.

\section*{Acknowledgments}
This work was conducted within the France 2030 framework programme, Centre Henri Lebesgue ANR-11-LABX-0020-01. This work has been partially supported by CNRS International Research Project Spectral Analysis of Dirac Operators – SPEDO. L.~L.T. acknowledges ChatGPT/Claude 3.5 Sonnet for assistance in refining the English and developing Python scripts for visualizations.

\appendix

\section{Proof of Proposition \ref{prop.straightstrip}}\label{sec.proofSS}
\subsection{Proof of Point \ref{pt.2xi0} of Proposition \ref{prop.straightstrip}}
This operator $\mathfrak{d}_{0,0,\mu}$  plays a privileged role in our study. 

Let us divide the proof into several steps.
Let \( \lambda \in \mathbb{R} \) denote an eigenvalue with \( \varphi \) as a corresponding eigenvector of the operator \( \mathfrak{d}_{0,0,\mu}\). Given that \( \left(\mathfrak{d}_{0,0,\mu}\right)^2= -\partial_t^2+ \mu^2 \), it follows that 
\[
\left(\partial^2_t + \lambda^2 - \mu^2\right)\varphi = 0.
\]
To compute the solution, let us introduce a complex number \( k \in \mathbb{C} \) satisfying the relation \( k^2 = \lambda^2 - \mu^2 \).

\begin{description}[style=unboxed,leftmargin=0cm]
\item[Case $k = 0$] 

We have \( \varphi\colon t \mapsto A + tB \), where \( A,B \in \mathbb{C}^2 \). Imposing the boundary conditions results in the system
\[
\begin{aligned}
    -\sigma_1 (A+B) &= A+B, \\
    \sigma_1 (A-B) &= A-B,
\end{aligned}
\]
which implies that \( B = -\sigma_1 A \). Consequently, \( \varphi \) simplifies to \( \varphi\colon t \mapsto (1_2 - \sigma_1 t)A \). Using the equation \( \mathfrak{d}_{0,0,\mu}\varphi - \lambda \varphi = 0 \), we obtain
\[
\begin{aligned}
    (\mu\sigma_3 + i\sigma_2\sigma_1 - \lambda I_2)A &= 0, \\
    -(\mu\sigma_3 - \lambda I_2)\sigma_1 A &= 0,
\end{aligned}
\]
and further simplification leads to \( \mu = -\frac{1}{2} \), \( \lambda = \pm\frac{1}{2} \) and
$
    ( \mu\sigma_3 + \lambda1_2)A = 0.
%
$
Considering \(\lambda = \frac{1}{2}\), we obtain that
\[
A = c e_1,
\]
where \(c \in \mathbb{C} \setminus \{0\}\) and \(e_1, e_2\). 
Therefore, \(u_+\) defined as 
\[
u_+ \colon t \mapsto \sqrt{\frac{3}{8}} \begin{pmatrix}  1 \\ -t \end{pmatrix}
\]
is a normalized eigenvector of \(\mathfrak{d}_{0,0,\mu}\) associated with the eigenvalue \(\frac{1}{2}\). Note that
\[
u_- \colon t \mapsto \sigma_1 u_+(t) = \sqrt{\frac{3}{8}} \begin{pmatrix} -t \\  1 \end{pmatrix},
\]
is a normalized eigenvector of \(\mathfrak{d}_{0,0,\mu}\) associated with the eigenvalue \(-\frac{1}{2}\).
\item[Case $k\ne0$]

We have \( \varphi\colon t \mapsto A\cos(k t)+B\sin(k t) \) with $A,B\in\mathbb{C}^2$.
The boundary conditions read
\[
\begin{aligned}
    -\sigma_1(A\cos(k) + B\sin(k))&= A\cos(k) + B\sin(k), \\
    \sigma_1(A\cos(k) - B\sin(k))&= A\cos(k) - B\sin(k),
\end{aligned}
\]
which simplifies as 
\[
A\cos(k) = -\sigma_1B\sin(k)\,.
\]
Applying the equation \(\mathfrak{d}_{0,0,\mu}\varphi - \lambda \varphi = 0 \), we obtain
\[
\begin{aligned}
    (\mu\sigma_3 - \lambda 1_2)A -i\sigma_2 k B&= 0, \\
    (\mu\sigma_3 - \lambda1_2)B + i\sigma_2k A &= 0.
\end{aligned}
\]
Given that \( k^2 \neq 0 \) (so that $\lambda\ne \mu$),  we deduce from these equations that \( A = 0 \) and \( B = 0 \) are equivalent. These conditions would lead to a contradiction, as they would not permit non-trivial solutions for the eigenvectors. Hence, we conclude that both \( A \) and  \( B \) must be non-zero. By the boundary condition, this implies that \( \sin(k) \), \( \cos(k) \) and \( \sin(2k) \) must also be non-zero.
Note that since  $k\in (i\R\cup \R)$, the numbers $k\tan(k)$ and $\frac{k}{\sin(2k)}$ are real.
We can reformulate the previous equations as
\[
\begin{aligned}
    A + \sigma_1 B \tan(k) &= 0\,, \\
    \left((\mu-k\tan(k))\sigma_3 -  \lambda1_2\right)B&=0\,.
\end{aligned}
\]
The second equation has a non-trivial solution if and only if \( (k\tan(k) - \mu)^2 = \lambda^2 \) which simplifies to
\begin{equation}\label{eq.mu-k}
\mu 
= -\frac{k}{\tan(2k)}\,.
\end{equation}
Note also that
\(
\mu - k\tan(k) = -k\left(\frac{1}{\tan(2k)} +\tan(k)\right) =-\frac{k}{\sin(2k)}\,,
\)
and
\(
\lambda^2 = \mu^2 + k^2 = k^2\left(1+\frac{1}{\tan^2(2k)}\right) = \frac{k^2}{\sin^2(2k)}
\)
so that
\[
B\text{ is colinear to }\begin{cases}
	e_2&\text{ if }\lambda\frac{k}{\sin(2k)}>0\,,
	\\
	e_1&\text{ if }\lambda\frac{k}{\sin(2k)}<0\,.
\end{cases}
\]
We get then that
\[
\varphi\colon t\longmapsto c\begin{cases}
	\frac{\cos(kt)}{\cos(k)}e_1 - \frac{\sin(kt)}{\sin(k)}e_2&\text{ if }\lambda\frac{k}{\sin(2k)}>0\,,
	\\
	\frac{\cos(kt)}{\cos(k)}e_2 - \frac{\sin(kt)}{\sin(k)}e_1&\text{ if }\lambda\frac{k}{\sin(2k)}<0\,,
\end{cases}
\]
is a normalized eigenvector when
\[
	c = \left(\frac{k\sin^2(2k)}{4k-\sin(4k)}\right)^{\frac{1}{2}}.
\]
Note that in the limit $k\to 0$ for $\lambda>0$, $\varphi$ tends to the eigenvector $u_+$  from the case $k = 0$.
\item[Study of equation \eqref{eq.mu-k} when $k^2<0$]
then $k = i\tilde k$ with  $\tilde k\in \R$ and \eqref{eq.mu-k} reduces to
\[
\mu = -\frac{\tilde k}{\tanh(2\tilde k)}.
\]
Examining the image of the function $g\colon \tilde k\mapsto -\frac{\tilde k}{\tanh(2\tilde k)}$, we obtain that $\mu\in(-\infty,-1/2]$ and that for such $\mu$, there exists a unique non negative solution $\tilde k$ (the unique non positive solutions being $-\tilde k$). We denote by
\[
	\begin{array}{llll}
		\tilde k\colon &(-\infty,-1/2]&\longrightarrow&[0,+\infty)
	\end{array}
\]
the inverse map and remark that $\lambda^2 = \mu^2-\tilde k(\mu)^2>0$ for all $\mu\in(-\infty,-1/2]$ and
\[
	\lim_{\mu\to-\infty}\mu^2-\tilde k(\mu)^2 = 0\,,
\]
see Figure \ref{fig:waveguide24}.
\begin{figure}[htbp]
    \centering
    \includegraphics[width=1\textwidth]{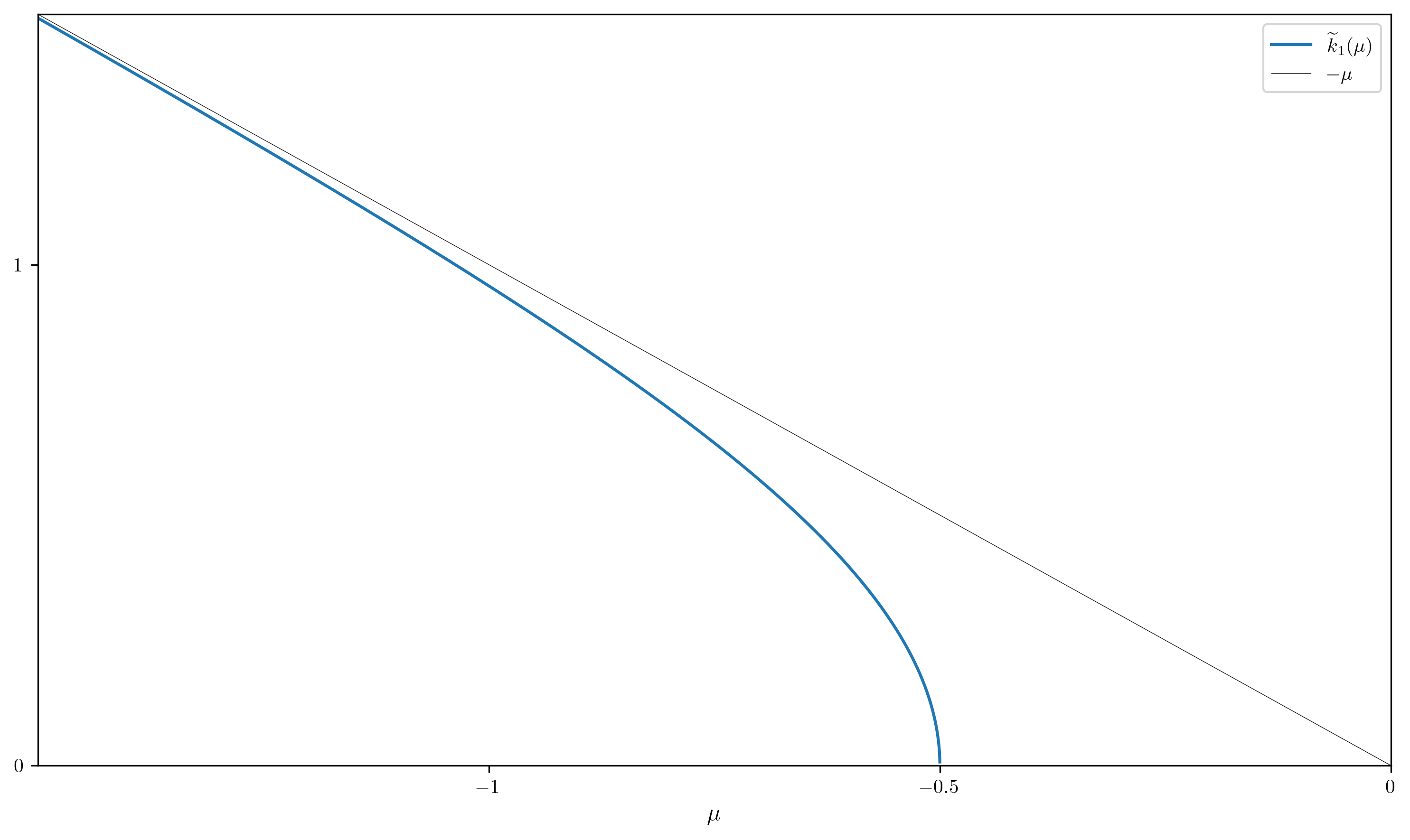}
    \caption{Plot $\mu\mapsto\tilde k_1(\mu)$.}
    \label{fig:waveguide24}
\end{figure}
\item[Study of equation \eqref{eq.mu-k} when $k^2>0$]
  then $k\in \R$. Examining the image of the function $ k\mapsto -\frac{ k}{\tan(2k)}$, we obtain that there exists infinitely many non negative solutions $k$ (the non positive solutions being $-k$). For $\mu>-\frac{1}{2}$, the increasing sequence of the positive solutions to this equation is $(k_n(\mu))_{n\geq 1}$. $(k_n(\mu))_{n\geq 2}$ can be analytically extended to $\mu\in\mathbb{R}$.
Figure \ref{fig:dirac1deig} illustrates the behavior of the curves \( k_1 \), \( k_2 \), and \( k_3 \) along with their respective asymptotic tendencies. Notably, \( k_1 \) vanishes at \( x = -0.5 \). As we traverse towards infinity, \( k_1 \) aligns with \( \frac{\pi}{2} \), while \( k_2 \) and \( k_3 \) respectively approach \( \pi \) and \( \frac{3\pi}{2} \).

\begin{figure}[htbp]
    \centering
    \includegraphics[width=1\textwidth]{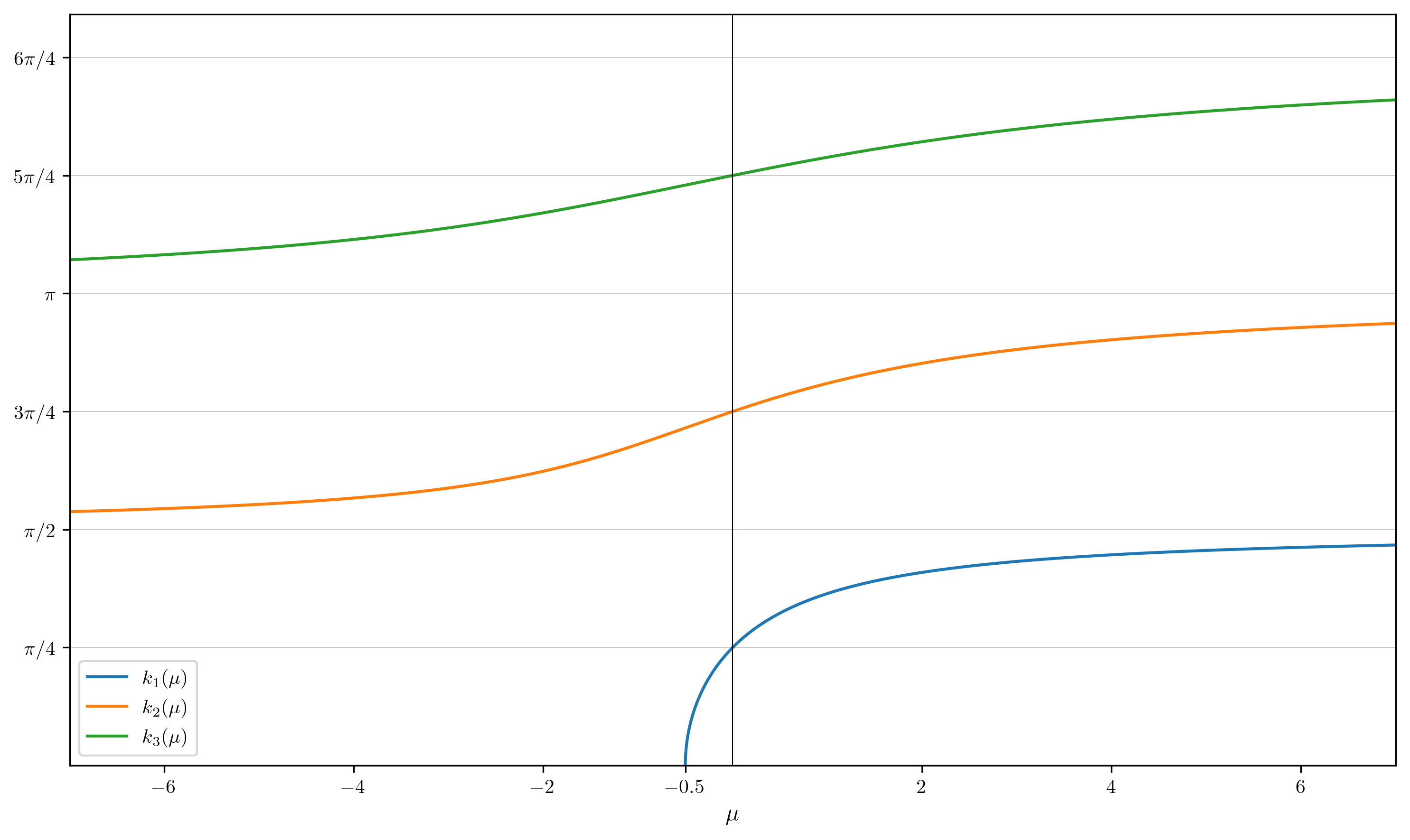}
    \caption{Curves representing \(k_1\), \(k_2\), and \(k_3\) for the one-dimensional Dirac operator.}
    \label{fig:dirac1deig}
\end{figure}

\end{description}

\subsection{Proof of Point \ref{pt.3xi0} of Proposition \ref{prop.straightstrip}}

The operator \(\mathfrak{d}_{0,\xi,\mu} \) is self-adjoint, possesses a discrete real spectrum and satisfies the following intertwining relationship:
\[
	\sigma_1\mathfrak{d}_{0,\xi,\mu}= -\mathfrak{d}_{0,-\xi,\mu}\sigma_1, \quad \sigma_1(\mathrm{Dom}(\mathfrak{d}_{0,\xi,\mu})) = \mathrm{Dom}(\mathfrak{d}_{0,\xi,\mu}) = \mathrm{Dom}(\mathfrak{d}_{0,0,\mu}),
\]
from which it follows that
\[
	\mathrm{sp}(\mathfrak{d}_{0,\xi,\mu}) = -\mathrm{sp}(\mathfrak{d}_{0,-\xi,\mu}).
\]
The previous relationships can be linked to the supersymmetric theory of Dirac operator for which one can find an introduction in \cite[Chapter 5]{MR1219537}. Under this formalism, $\sigma_1$ is called a \emph{grading operator} and $\mathfrak{d}_{0,0,\mu}$ a \emph{supercharge}. Applying \cite[Theorem 5.13]{MR1219537} to $\mathfrak{d}_{0,\xi,\mu} = \mathfrak{d}_{0,0,\mu}+\sigma_1\xi$, we deduce that 
\begin{equation}\label{eq.specxi}
	\mathrm{sp}(\mathfrak{d}_{0,\xi,\mu}) 
	= \left\{
		\pm\sqrt{\xi^2 + \lambda^2}\,,\lambda\in\mathrm{sp}(\mathfrak{d}_{0,0,\mu})
	\right\}\,,
\end{equation}
and Point \ref{pt.1xi0} follows.
Let us give some details of the arguments leading to this result as well as some explicit formulas for the eigenvectors.

Let $\lambda\in\mathrm{sp}(\mathfrak{d}_{0,0,\mu})\cap \mathbb{R}^+$ and $\varphi_+\in  \mathrm{Dom}(\mathfrak{d}_{0,0,\mu})$ an associated normalized eigenvector. We remark that $\varphi_- = \sigma_1 \varphi_+$ is a normalized eigenvector of $\mathfrak{d}_{0,0,\mu}$ associated with the eigenvalue $-\lambda$. Since $\lambda\ne0$, $\mathcal{B} = \{\varphi_\pm\}$ is an orthonormal family that spans a vector space left stable by $\mathfrak{d}_{0,0,\mu}$, $\sigma_1$ and $\mathfrak{d}_{0,\xi,\mu}$. The matrix of the operator $\mathfrak{d}_{0,\xi,\mu}$ in the basis $\mathcal{B}$ is then
\[
	\xi\sigma_1 + \lambda\sigma_3\,,
\]
whose eigenvalues are $\{\pm\sqrt{\xi^2 + \lambda^2}\}$. This shows \eqref{eq.specxi}. We can also verify that the vector whose coordinates in $\mathcal{B}$ are
\[
v_{\xi,+} = c_\xi\begin{pmatrix}
	1\\
	\frac{\xi}{\sqrt{\xi^2+\lambda^2} + \lambda}
\end{pmatrix}\,
\text{ with }
c_\xi = \left(\frac{\sqrt{\xi^2+\lambda^2} + \lambda}{2\sqrt{\xi^2+\lambda^2}}\right)^{\frac{1}{2}}\,,
\]
is a normalized eigenvector of $\mathfrak{d}_{0,\xi,\mu}$ associated with the eigenvalue $\sqrt{\xi^2 + \lambda^2}$. 

%

%
%
%
%

\subsection{Analytic expansions}
The Taylor series expansions of the functions $k_1$ and $\nu_1(0,\mu)$ at $0$ begin as:
\[\begin{split}
k_1(\mu) 
&=
\frac{\pi}{4}
+\frac{2}{\pi} \mu 
- \frac{16}{\pi^3} \mu^2 
+ \left(\frac{256}{\pi^5} - \frac{32}{3 \pi^3}\right)\mu^3
+\left(-\frac{5120}{\pi^7} + \frac{1024}{3 \pi^5}\right)\mu^4
+  \mathcal{O}(\mu^5)
\,,
\\
\nu_1(0,\mu)
&=
\frac{\pi}{4} 
+ \frac{2}{\pi}  \mu
+ \left(- \frac{16}{\pi^{3}} + \frac{2}{\pi^{}}\right) \mu^{2} 
+ \left(\frac{256}{\pi^5} - \frac{80}{3 \pi^3}\right)\mu^{3}  
+ \left(- \frac{5120}{\pi^{7}} - \frac{8}{\pi^{3}} + \frac{1792}{3 \pi^{5}}\right)\mu^{4}
+ O\left(\mu^{5}\right)
\,.
\end{split}\]
We computed these coefficients using a SymPy code based on the fact that \( k_1 \) solves
\[
\begin{cases}
	\frac{dk_1(\mu)}{d\mu} = \frac{k_1(\mu)}{\mu + 2\mu^2 + 2k_1(\mu)^2}, \\
	k_1(0) = \frac{\pi}{4}.
\end{cases}
\]
The code solves the equations for the coefficients of the series expansion of $k_1$ obtained by equating the expansions of \( k_1' \) and \(\frac{k_1}{\mu + 2\mu^2 + 2k_1^2}\).

\subsection{Momentum computation}
Along the proof, we will need the following result.
\begin{lemma}\label{lem.moment}
Let $\xi, \mu\in\R$. We have
\[
\langle\varphi_{\xi,\mu,1}, \sigma_1 t \varphi_{\xi,\mu,1}\rangle
=
\frac{4}{\pi^2} +(|\xi|^2+|\mu|)r\,,
\]
with $r\in S^0(\R^2, \mathscr{L}(\C))$.
\end{lemma}
\begin{proof}
By Proposition \ref{prop.straightstrip}, we have
\[\begin{split}
\langle\varphi_{\xi,\mu,1}, \sigma_1 t \varphi_{\xi,\mu,1}\rangle
=|c_\xi|^2\left(
(1+|c^1_\xi|^2)\langle\varphi_{0,\mu,1}, \sigma_1 t \varphi_{0,\mu,1}\rangle
+2c^1_\xi \langle\varphi_{0,\mu,1},  t \varphi_{0,\mu,1}\rangle
\right)\,,
\end{split}\]
with $c_\xi^1 = \frac{\xi}{\sqrt{\xi^2+\nu(0,\mu)^2} + \nu(0,\mu)}$. Using the explicit formula for $\varphi_{0,\mu,1}$, we get that 
\[\begin{split}
	&\langle\varphi_{0,\mu,1},  t \varphi_{0,\mu,1}\rangle =0\,,
	\\&\langle\varphi_{0,\mu,1},  \sigma_1t \varphi_{0,\mu,1}\rangle =2|c|^2\int_{-1}^1 t\left(
	\frac{\sin(kt)\cos(kt)}{\sin(k)\cos(k)}
	\right){\rm d}t
	=
	2|c|^2\frac{1 - 2k \cot(2k)}{2k^2}\,,
\end{split}\]
where $c = \left(\frac{k\sin^2(2k)}{4k-\sin(4k)}\right)^{\frac{1}{2}}$. The result follows from Proposition \ref{prop.straightstrip}, explicit computations and
\[\begin{split}
	&\langle\varphi_{0,\mu,1},  \sigma_1t \varphi_{0,\mu,1}\rangle = \frac{4}{\pi^2} + \mathscr{O}(|\mu|)\,,
	\\
	&|c_\xi|^2(1+|c^1_\xi|^2) = 1 + \mathscr{O}(|\xi|^2)\,.
\end{split}\]
\end{proof}

\section{Some technical lemmas}

\subsection{The resolvent as a symbol}
We recall that $\mathfrak{d}_0$ is defined in \eqref{eqn:defdbarj}, $\mathfrak{d}^\pm_0$ in \eqref{eq.d0pm} and $\Xi_{1,\pm,\epsilon}$ in \eqref{eq.regu111}.

\begin{lemma}\label{lem.A1}
Let $M>0$. There exists $\epsilon_0,C>0$ such that for all $\mu\in(-1,1)$, all $\lambda\in[0,\nu_1(0,\mu)+M\epsilon^2]$,  and all $\epsilon\in(0,\epsilon_0)$, the following holds.
\begin{enumerate}[label = \rm (\roman*)]
\item
For  all $\xi\in\R$, the operator ${\mathfrak{d}^\pm_0} - \lambda$ is bijective and 
	\begin{equation}\label{eq.estires1}
		\|(\mathfrak{d}^\pm_0-\lambda)^{-1}\|\leq C\epsilon^{-2\eta}\langle\xi\rangle^{-1} \,.
	\end{equation}
	\item \label{pt.symbolD0pm} The operator $({\mathfrak{D}_0^\pm} - \lambda)$ is invertible and $({\mathfrak{D}_0^\pm} - \lambda)^{-1} = \left(({\mathfrak{d}_0^\pm} - \lambda)^{-1}\right)^{W}$ is a Fourier multiplier whose symbol $({\mathfrak{d}_0^\pm} - \lambda)^{-1}$ belongs to $\epsilon^{-2\eta}S^{-1}_{2\eta}(\R^2,\mathscr{L}(\mathscr{B}))$ and
	\begin{equation}\label{eq.estires2}
		\|({\mathfrak{D}_0^\pm} - \lambda)^{-1}\|\leq C\epsilon^{-2\eta}\,.
	\end{equation}
\end{enumerate}
\end{lemma}
	\begin{proof}
	Let us start by noting that there exists a constant $c>0$ such that for $\mu\in(-1,1)$, 
	\[
		c\leq \nu_1(0,\mu)= \sqrt{\mu^2+k_1(\mu)}\leq c^{-1}\,,
	\]
	where the $\nu_n$ are given in Proposition \ref{prop.straightstrip}.
	The proof is structured into several steps.
	\begin{enumerate}[label=\it Step \arabic*.,  leftmargin=0cm, itemindent = 1.5cm]
	
		\item First, we establish a lower bound on $\nu_1(\Xi_{1,\pm,\epsilon}(\xi),\mu)-\lambda$. We have for all $\xi\in\mathbb{R}^*$,
		\[\begin{split}
			&\nu_1(\xi,\mu)-\lambda
			\geq 
			\frac{\nu_1(\xi,\mu)^2-\lambda^2}{\nu_1(\xi,\mu)+\lambda}
			\geq 
			\frac{\nu_1(\xi,\mu)^2-(\nu_1(0,\mu)+M\epsilon^2)^2}{\nu_1(\xi,\mu)+\nu_1(0,\mu)+M\epsilon^2}
			\\&\geq
			\frac{\xi^2+\mu^2+k_1(\mu)-(\sqrt{\mu^2+k_1(\mu)}+M\epsilon^2)^2}{\sqrt{\xi^2+ \mu^2+k_1(\mu)}+\sqrt{\mu^2+k_1(\mu)}+M\epsilon^2}
			\\&\geq
						\frac{
				\xi^2-2\sqrt{\mu^2+k_1(\mu)}M\epsilon^2 - M^2\epsilon^4
			}{
				\sqrt{\xi^2+ \mu^2+k_1(\mu)}+\sqrt{\mu^2+k_1(\mu)}+M\epsilon^2
			}
			\geq\langle\xi\rangle C_\xi\,,
		\end{split}\]
		where, for some constant $C>0$ uniform in $\epsilon$ and $\xi$, we have set
		\[
			C_\xi = \frac{|\xi|}{\langle\xi\rangle}
			\frac{
				1-C\epsilon^2/\xi^2
			}{
				\sqrt{1+ \frac{\mu^2+k_1(\mu)}{\xi^2}}+\sqrt{\mu^2+k_1(\mu)}/ |\xi|+M\epsilon^2/|\xi|
			}\,.
		\]
		The map $\xi\mapsto C_\xi$  is increasing on $[\sqrt{C}\epsilon,+\infty)$ and even. Moreover, as $\eta \in (0,\frac12)$ for $\epsilon$ small enough, one has
		\[
			|\Xi_{1,\pm,\epsilon}|(\R) = \left[\epsilon^\eta \frac{R}2,+\infty\right) \subset [\sqrt{C}\epsilon,+\infty)
		\]
		so that $C_{\Xi_{1,\pm,\epsilon}(\xi)}\geq  C_{\epsilon^\eta R/2}$. A straightforward computation then shows that for some constant $k > 0$ there holds $C_{\epsilon^\eta R/2}\geq k\epsilon^{2\eta}$, which proves that
		\begin{equation}\label{eq.complowerboundnu1}
			\nu_1(\Xi_{1,\pm,\epsilon}(\xi),\mu)-\lambda\geq  \langle\xi\rangle k \epsilon^{2\eta}\,.
		\end{equation}
		\item The spectrum of $\mathfrak{d}_0^\pm$ consists of the values  $(\pm \nu_n(\Xi_{1,\pm,\epsilon}(\xi),\mu))_{n\geq 1}$ ensuring that  \eqref{eq.complowerboundnu1} makes ${\mathfrak{d}^\pm_0} - \lambda$ bijective and 
		\begin{equation}\label{eq.complowerboundnu2}
			\|({\mathfrak{d}_0^\pm} - \lambda)^{-1}\|_{\mathscr{L}(\mathscr{B})} = \frac1{{\rm dist}(\lambda,\rm sp({\mathfrak{d}_0^\pm}))}=\frac{1}{\nu_1(\Xi_{1,\pm,\epsilon}(\xi),\mu)-\lambda}\leq \frac1k\epsilon^{-2\eta}\langle\xi\rangle^{-1}\,.
		\end{equation}
		This implies \eqref{eq.estires1} and \eqref{eq.estires2}
		\item
		Differentiating the identity $({\mathfrak{d}_0} - \lambda)({\mathfrak{d}_0} - \lambda)^{-1} = {\rm Id}$, with respect to $\xi$ gives
		\begin{multline*}
		0 
		= \left(\partial_\xi({\mathfrak{d}_0} - \lambda)\right)({\mathfrak{d}_0} - \lambda)^{-1} 
		+({\mathfrak{d}_0} - \lambda)\partial_\xi({\mathfrak{d}_0} - \lambda)^{-1}
		= \sigma_1({\mathfrak{d}_0} - \lambda)^{-1} +({\mathfrak{d}_0} - \lambda)\partial_\xi({\mathfrak{d}_0} - \lambda)^{-1}
		\end{multline*}
		so that
		\begin{equation*}
			\partial_\xi({\mathfrak{d}_0} - \lambda)^{-1} = -({\mathfrak{d}_0} - \lambda)^{-1}\sigma_1({\mathfrak{d}_0} - \lambda)^{-1}\,,
		\end{equation*}
		and by induction, we get for all $k\in\mathbb{N}$,
		\begin{equation}\label{eq.idRes}
			\partial_\xi^{(k)}({\mathfrak{d}_0} - \lambda)^{-1} = (-1)^kk!\left(({\mathfrak{d}_0} - \lambda)^{-1}\sigma_1\right)^{ k}({\mathfrak{d}_0} - \lambda)^{-1}\,.
		\end{equation}
		\item In this step, we examine properties of the derivatives of  $\Xi_{1,\pm,\epsilon}$. In particular, one remarks that there exist constants $(C_\alpha)_{\alpha\in\mathbb{N}}\subset(0,+\infty)$ such that for all $\xi\in\mathbb{R}$, all $\epsilon\in(0,\epsilon_0)$, all $\alpha\in\mathbb{N}^*$
		\begin{equation}\label{eq.propXipm}
			\begin{split}
				&|\Xi_{1,\pm,\epsilon}(\xi)|\leq C_0\langle\xi\rangle\,,
				\\&
				|\partial_\xi^{(\alpha)}\Xi_{1,\pm,\epsilon}(\xi)| 
				= \epsilon^{\eta(1-\alpha)}|\Xi_{1,\pm}^{(\alpha)}(\epsilon^{-\eta}\xi)|
				\leq C_\alpha\epsilon^{\eta(1-\alpha)}\,.
			\end{split}
		\end{equation}
		\item To obtain Point \ref{pt.symbolD0pm}, we will now combine identities \eqref{eq.estires1}, \eqref{eq.idRes} and \eqref{eq.propXipm} with Faà-Di-Bruno's formula:  we have for all $\xi\in\mathbb{R}$, $\alpha\in\mathbb{N}$,
		\[\begin{split}
			&\partial_\xi^{(\alpha)} (\mathfrak{d}^\pm_0-\lambda)^{-1} 
			=  \partial_\xi^{(\alpha)}\left( (\mathfrak{d}_0-\lambda)^{-1}\circ \Xi_{1,\pm,\epsilon}(\xi)\right)
			\\&= \sum_{k=1}^{\alpha}\left( (\mathfrak{d}_0-\lambda)^{-1}\right)^{(k)}(\Xi_{1,\pm,\epsilon}(\xi)) \cdot B_{\alpha,k} \left(\Xi_{1,\pm,\epsilon}'(\xi), \Xi_{1,\pm,\epsilon}''(\xi), \ldots, \Xi_{1,\pm,\epsilon}^{(\alpha-k+1)}(\xi)\right),
		\end{split}\]
where \(B_{\alpha,k}\) are the Bell polynomials defined as:
\[
B_{\alpha,k}(x_1, x_2, \ldots, x_{\alpha-k+1}) = \sum \frac{\alpha!}{j_1! j_2! \cdots j_{\alpha-k+1}!} \left( \frac{x_1}{1!} \right)^{j_1} \left( \frac{x_2}{2!} \right)^{j_2} \cdots \left( \frac{x_{\alpha-k+1}}{(\alpha-k+1)!} \right)^{j_{\alpha-k+1}},
\]
and the sum is taken over all sequences \(j_1, j_2, \ldots, j_{\alpha-k+1}\) of nonnegative integers such that:
\[
j_1 + j_2 + \cdots + j_{\alpha-k+1} = k \quad \text{and} \quad j_1 + 2j_2 + \cdots + (\alpha-k+1)j_{\alpha-k+1} = \alpha.
\]
For $k\in\{1,\dots,\alpha\}$, and a sequence \(j_1, j_2, \ldots, j_{\alpha-k+1}\) as above,  by \eqref{eq.propXipm},
 there exists $C>0$ such that
\[\begin{split}
	&\left|
		\left( \partial_\xi^{(1)}\Xi_{1,\pm,\epsilon}(\xi) \right)^{j_1} \left(\partial_\xi^{(2)}\Xi_{1,\pm,\epsilon}(\xi) \right)^{j_2} \cdots \left( \partial_\xi^{(\alpha-k+1)}\Xi_{1,\pm,\epsilon}(\xi) \right)^{j_{\alpha-k+1}}
	\right|
	\\&\leq
	C \left[\epsilon^\eta\right]^{\sum_{i = 1}^{\alpha-k+1}j_i(1-i)}
	 = C\left[\epsilon^\eta\right]^{k - \alpha}\,,
\end{split}\]
so that
\begin{equation}\label{eq.Bell}
|B_{\alpha,k} \left(\Xi_{1,\pm,\epsilon}'(\xi), \Xi_{1,\pm,\epsilon}''(\xi), \ldots, \Xi_{1,\pm,\epsilon}^{(\alpha-k+1)}(\xi)\right)|\leq C\epsilon^{\eta(k - \alpha)}\,.
\end{equation}
By \eqref{eq.estires1}, \eqref{eq.complowerboundnu2} and \eqref{eq.idRes}, there exists $C>0$ such that
\begin{equation}\label{eq.FDB1}
	\|\left( (\mathfrak{d}_0-\lambda)^{-1}\right)^{(k)}(\Xi_{1,\pm,\epsilon}(\xi)) \|_{\mathscr{L}(\mathscr{B})}
	\leq
	C\epsilon^{-2\eta(k+1)}<\xi>^{-(k+1)}\,.
\end{equation}
Hence, by \eqref{eq.Bell} and \eqref{eq.FDB1}, we get
\[\begin{split}
	&\left\|
		\left( (\mathfrak{d}_0-\lambda)^{-1}\right)^{(k)}(\Xi_{1,\pm,\epsilon}(\xi)) \cdot B_{\alpha,k} \left(\Xi_{1,\pm,\epsilon}'(\xi), \Xi_{1,\pm,\epsilon}''(\xi), \ldots, \Xi_{1,\pm,\epsilon}^{(\alpha-k+1)}(\xi)\right)
	\right\|_{\mathscr{L}(\mathscr{B})}
	\\&\leq C\left[\epsilon^\eta\right]^{-2(k+1)+k-\alpha}\langle\xi\rangle^{-(k+1)}
	\leq C\left[\epsilon^\eta\right]^{-\alpha-2 - k}\langle\xi\rangle^{-1}
\end{split}\]
Therefore, summing over $k$, we get
\begin{align*}
	\left\|
	\partial_\xi^{(\alpha)} (\mathfrak{d}^\pm_0-\lambda)^{-1}
	\right\|_{\mathscr{L}(\mathscr{B})}
	&\leq
	C\left[\epsilon^{-\eta}\right]^{2+\alpha}\left(\sum_{k=1}^\alpha(\epsilon^{-\eta})^{k}\right)\langle\xi\rangle^{-1}\,\\
	& \leq C' (\epsilon^{-\eta})^{2(1+\alpha)}\langle\xi\rangle^{-1}
\end{align*}
	We then deduce  that $({\mathfrak{d}_0^\pm} - \lambda)^{-1} \in \epsilon^{-2\eta}S^{-1}_{2\eta}(\R^2,\mathscr{L}(\mathscr{B}))$.
	\end{enumerate}
\end{proof}


%
%

\subsection{The orthogonal projector}\label{sec.symbolProj}
	We recall that $\Pi_{\xi,\mu}=\langle\varphi_{\xi,\mu,1},\cdot\rangle_{L^2(I)}$ is defined in Lemma \ref{lem.parametrix} and $\varphi_{\xi,\mu,1}$ in Proposition \ref{prop.straightstrip}.
\begin{lemma}\label{lem.A3}
	For $\mathscr{X} \in \{\mathscr{A},\mathscr{B}\}$, the operator $\Pi_{\xi,\mu}$ belongs to $S^0(\R^2,\mathscr{L}(\mathscr{X},\mathbb{C}))$ and the operator $\Pi_{\xi,\mu}^*$ belongs to $S^0(\R^2,\mathscr{L}(\mathbb{C},\mathscr{X}))$. In particular, $\Pi^*_{\xi,\mu}\Pi_{\xi,\mu}$ belongs to $S^0(\R^2,\mathscr{L}(\mathscr{B},\mathscr{A}))$.
\end{lemma}
\begin{proof}

In this proof, we adhere to the notation established in Proposition \ref{prop.straightstrip}. 

The mappings $\xi \mapsto c_\xi$ and $\xi \mapsto \frac{\xi}{\sqrt{\xi^2+\nu(0,\mu)^2} + \nu(0,\mu)}$ are elements of $S^0(\mathbb{R}^2,\mathscr{L}(\mathscr{B},\mathscr{B}))$. The function $\varphi_{0,\mu,1}$ is an element of $\mathscr{A}$. 
Consequently, by Theorem \ref{thm:comp_thm}, the following statements hold:
\begin{enumerate}
    \item The operator $\Pi_{\xi,\mu}^*$ is an element of $S^0(\mathbb{R}^2,\mathscr{L}(\mathbb{C},\mathscr{A})) \subset S^0(\mathbb{R}^2,\mathscr{L}(\mathbb{C},\mathscr{B}))$.
    \item The operator $\Pi_{\xi,\mu}$ is an element of $S^0(\mathbb{R}^2,\mathscr{L}(\mathscr{B},\mathbb{C}))\subset S^0(\mathbb{R}^2,\mathscr{L}(\mathscr{A},\mathbb{C})$.
\end{enumerate}
\end{proof}

\begin{lemma}\label{lem.A4}
There exist $\delta,\mu_0>0$ such that for all $z\in(0,\frac{\pi}{4}+\delta)$, all $\mu\in(-\mu_0,\mu_0)$ and all $\xi\in\mathbb{R}$, the operator
\[(\mathfrak{d}_0 - z)_\perp:=(\mathfrak{d}_0 - z)(\mathrm{Id}-\Pi_{\xi,\mu}^*\Pi_{\xi,\mu})\]
is inversible when seen as an endomorphism of $\{\varphi_{\xi,\mu,1}\}^\perp$ and the symbol of 
\[
0\oplus(\mathfrak{d}_0 - z)_\perp^{-1}\colon
\mathscr{B} = {\rm span }\{\varphi_{\xi,\mu,1}\}\oplus \{\varphi_{\xi,\mu,1}\}^\perp\longrightarrow {\rm span}\{\varphi_{\xi,\mu,1}\}\oplus \{\varphi_{\xi,\mu,1}\}^\perp
\]
 is an element of $S^{-1}(\R^2,\mathscr{L}(\mathscr{B},\mathscr{B}))\cap S^{0}(\R^2,\mathscr{L}(\mathscr{B},\mathscr{A}))$.
\end{lemma}
\begin{proof}
The proof is divided into several parts.
\begin{enumerate}[label=\it Step \arabic*.,  leftmargin=0cm, itemindent = 1.5cm]
\item
Let $\delta_1 := \min_{|\mu|\leq 1}|\nu_1(0,\mu)-\nu_2(0,\mu)|>0$. By continuity, there exists $\mu_0\in(0,1)$ such that for all $|\mu|\leq \mu_0,$
\[
	|\nu_1(0,\mu)-\pi/4|\leq\delta_1/4\,.
\]
Let $z\in(0,\pi/4+ \delta_1/4)$ and $\mu\in(-\mu_0,\mu_0)$.
We have
\[
	\nu_1(0,\mu)-z
	\geq (\nu_1(0,\mu)-\pi/4) - \delta_1/4
	\geq - \delta_1/2\,,
\]
so that
\[
	\nu_2(0,\mu)-z
	= (\nu_2(0,\mu)-\nu_1(0,\mu)) + (\nu_1(0,\mu)-z)
	\geq
	\delta_1 - \delta_1/2 = \delta_1/2\,.
\]
We also have that for $\xi\in\R$,
\[
\nu_2(\xi,\mu)-\nu_2(0, \mu)
 =\frac{\xi^2}{\nu_2(\xi,\mu)-\nu_2(0, \mu)} 
 \geq \frac{\xi^2}{2\nu_2(\xi,\mu)} 
 \geq \frac{1}{2c_0}\xi^2\langle\xi\rangle^{-1}\,,
\] 
where
\[
	c_0 = \max\{1,(\nu_2(0,\mu))_{|\mu|\leq 1}\}>0\,,
\]
so that
\begin{equation}\label{eq.eqnu2}
	\nu_2(\xi,\mu)-z
	\geq \frac{1}{2c_0}\xi^2\langle\xi\rangle^{-1} + \frac{\delta_1}{2}
	\geq \min\left\{\frac{1}{2c_0}, \frac{\delta_1}{2}\right\}\left(\xi^2\langle\xi\rangle^{-1} + 1\right)
	\geq \min\left\{\frac{1}{2c_0}, \frac{\delta_1}{2}\right\}\langle\xi\rangle\,.
\end{equation}
We also have
\begin{equation}\label{eq.eqnu1}
	\nu_1(\xi,\mu)\geq \langle\xi\rangle\min\{1,(\nu_1(0,\mu))_{|\mu|\leq 1}\}\,.
\end{equation}
\item Let $(\varphi_{\xi,\mu,n})_{n\in\Z^*}$ be an orthonormal basis associated with the eigenvalues $({\rm sign}(n)\nu_{|n|}(\xi,\mu))_{n\in\Z^*}$ (this extends the conventions taken in Proposition \ref{prop.straightstrip}). We have 
\[(\mathfrak{d}_0 - z)_\perp=\bigoplus_{n\neq 0,1} (\nu_n(\xi,\mu)-z)|\varphi_{\xi,\mu,n}\rangle\langle\varphi_{\xi,\mu,n}| \,.\]
Therefore, the operator $(\mathfrak{d}_0 - z)_\perp^{-1}\colon \{\varphi_{\xi,\mu,1}\}^\perp\to\{\varphi_{\xi,\mu,1}\}^\perp$ is inversible and by \eqref{eq.eqnu2} and  \eqref{eq.eqnu1},
\begin{equation}\label{eq.d0perpub}
\|(\mathfrak{d}_0 - z)^{-1}_\perp\|
\leq \max\left(\frac{1}{\nu_2(\xi,\mu)-z},\frac{1}{\nu_{1}(\xi,\mu)+z}\right)
\leq \langle \xi\rangle^{-1} C_0\,,
\end{equation}
with $C_0 = \max\left(\frac{1}{\min\left\{\frac{1}{2c_0}, \frac{\delta_1}{2}\right\}},\frac{1}{\min\{1,(\nu_1(0,\mu))_{|\mu|\leq 1}\}}\right)>0$.
\item Let us now consider the derivatives. By taking the derivative of 
\[
\Pi_{\xi,\mu}^*\Pi_{\xi,\mu}\left(0\oplus(\mathfrak{d}_0 - z)_\perp^{-1}\right)=0\,,
\quad \text{ and }\quad 
(\mathfrak{d}_0-z)\left(0\oplus(\mathfrak{d}_0 - z)_\perp^{-1}\right)=\mathrm{Id}-\Pi^*_{\xi,\mu}\Pi_{\xi,\mu}\,,
\]
we get
\begin{equation}\label{eq.eq13}
\begin{split}
&\Pi_{\xi,\mu}^*\Pi_{\xi,\mu}\partial_\xi\left(0\oplus(\mathfrak{d}_0 - z)_\perp^{-1}\right)
=-\partial_\xi(\Pi_{\xi,\mu}^*\Pi_{\xi,\mu})\left(0\oplus(\mathfrak{d}_0 - z)_\perp^{-1}\right)
\,,
\end{split}\end{equation}
and
\begin{equation}\label{eq.eq14}
(\mathfrak{d}_0-z)\partial_\xi\left(0\oplus(\mathfrak{d}_0 - z)_\perp^{-1}\right)
=
-\sigma_1\left(0\oplus(\mathfrak{d}_0 - z)_\perp^{-1}\right)
-\partial_{\xi}(\Pi^*_{\xi,\mu}\Pi_{\xi,\mu})\,.
\end{equation}
From
\[
({\rm Id} - \Pi_{\xi,\mu}^*\Pi_{\xi,\mu})(\mathfrak{d}_0-z) = \left(0\oplus(\mathfrak{d}_0 - z)_\perp\right)({\rm Id} - \Pi_{\xi,\mu}^*\Pi_{\xi,\mu})\,,
\]
and \eqref{eq.eq14}, we deduce that
\[\begin{split}
&
({\rm Id} - \Pi_{\xi,\mu}^*\Pi_{\xi,\mu})(\mathfrak{d}_0-z)\partial_\xi\left(0\oplus(\mathfrak{d}_0 - z)_\perp^{-1}\right)
\\&
=
\left(0\oplus(\mathfrak{d}_0 - z)_\perp\right)({\rm Id} - \Pi_{\xi,\mu}^*\Pi_{\xi,\mu})\partial_\xi\left(0\oplus(\mathfrak{d}_0 - z)_\perp^{-1}\right)
\\&=
-({\rm Id} - \Pi_{\xi,\mu}^*\Pi_{\xi,\mu})
\left(
	\sigma_1\left(0\oplus(\mathfrak{d}_0 - z)_\perp^{-1}\right)
	+\partial_{\xi}(\Pi^*_{\xi,\mu}\Pi_{\xi,\mu})
\right)
\end{split}\]
and
\begin{equation}\label{eq.eq15}
\begin{split}
({\rm Id} - \Pi_{\xi,\mu}^*\Pi_{\xi,\mu})\partial_\xi\left(0\oplus(\mathfrak{d}_0 - z)_\perp^{-1}\right)
=&
-\left(0\oplus(\mathfrak{d}_0 - z)_\perp^{-1}\right)
	\sigma_1\left(0\oplus(\mathfrak{d}_0 - z)_\perp^{-1}\right)\\
&-\left(0\oplus(\mathfrak{d}_0 - z)_\perp^{-1}\right)
	\partial_{\xi}(\Pi^*_{\xi,\mu}\Pi_{\xi,\mu})\,.
\end{split}\end{equation}
By \eqref{eq.eq13} and \eqref{eq.eq15}, we obtain
\begin{equation}\label{eq.eq16}
\begin{split}
\partial_\xi\left(0\oplus(\mathfrak{d}_0 - z)_\perp^{-1}\right)
=
&-\left(0\oplus(\mathfrak{d}_0 - z)_\perp^{-1}\right)
	\sigma_1\left(0\oplus(\mathfrak{d}_0 - z)_\perp^{-1}\right)
	\\
	&-
	\left(
	\partial_\xi(\Pi_{\xi,\mu}^*\Pi_{\xi,\mu})\left(0\oplus(\mathfrak{d}_0 - z)_\perp^{-1}\right)
	+
	\left(0\oplus(\mathfrak{d}_0 - z)_\perp^{-1}\right)
	\partial_{\xi}(\Pi^*_{\xi,\mu}\Pi_{\xi,\mu})
	\right)\,.
\end{split}\end{equation}
By \eqref{eq.d0perpub}, Lemma \ref{lem.A3}, and Theorem \ref{thm:comp_thm}, and there exists \(C_1 > 0\) such that
\begin{equation}\label{eq.eq17}
	\|\partial_\xi\left(0\oplus(\mathfrak{d}_0 - z)_\perp^{-1}\right)\|
	\leq C_1\langle\xi\rangle^{-1}\,.
\end{equation}
\item By \eqref{eq.eq16}, \eqref{eq.eq17}, Lemma \ref{lem.A3}, and an induction argument, for all \(k \geq 2\), there exists \(C_k > 0\) such that
\begin{equation*}
	\|\partial_\xi^k\left(0\oplus(\mathfrak{d}_0 - z)_\perp^{-1}\right)\|
	\leq C_k\langle\xi\rangle^{-1}\,,
\end{equation*}
so that
 \(
0 \oplus (\mathfrak{d}_0 - z)_\perp^{-1} \in S^{-1}(\mathbb{R}^2, \mathscr{L}(\mathscr{B}, \mathscr{B})).
\)
\item

In order to obtain the estimate involving the \(H^1\)-norm, we use the estimate \eqref{eq.d0perpub} and the triangle inequality to derive
	\begin{equation*}
		\begin{split}
		\|D_t\left(0\oplus(\mathfrak{d}_0 - z)_\perp^{-1}\right)\| &=	\|\sigma_2 D_t\left(0\oplus(\mathfrak{d}_0 - z)_\perp^{-1}\right)\| \\
		&\leq \|(\mathfrak{d}_0  - z)\left(0\oplus(\mathfrak{d}_0 - z)_\perp^{-1}\right)\|
		+\left(|\xi| +|\mu|+ |z|\right)\|\left(0\oplus(\mathfrak{d}_0 - z)_\perp^{-1}\right)\|
		\\&\leq 1 + C_0\langle\xi\rangle^{-1}\left(|\xi| +|\mu|+ |z|\right)
		\,.
	\end{split}
	\end{equation*}
Therefore, there exists $\tilde C_0>0$ such that 
	\begin{equation}\label{eq.boundedd0-z}
	\|\left(0\oplus(\mathfrak{d}_0 - z)_\perp^{-1}\right)\|_{\mathscr{L}(\mathscr{B},\mathscr{A})} \leq \tilde C_0\,.
	\end{equation}
The control of the derivatives follows from \eqref{eq.boundedd0-z}, \eqref{eq.eq16}, \eqref{eq.eq17}, Lemma \ref{lem.A3}, and an induction argument.
\end{enumerate}

\end{proof}

\bibliographystyle{abbrv}
\bibliography{le.treust-ourmieres-raymond}

\end{document}